\documentclass[11pt]{article}
\usepackage{amsmath}
\usepackage{amssymb}
\usepackage{amsthm}
\usepackage{graphicx}
\usepackage{varioref}
\newtheorem{theorem}{Theorem}[section]
\newtheorem{lemma}{Lemma}[section]
\newtheorem{claim}{Claim}[section]
\newtheorem{definition}[theorem]{Definition}

\usepackage{fancyhdr}
\usepackage{verbatim} 
\usepackage{multicol}
\usepackage{enumerate}
\usepackage[normalem]{ulem}
\usepackage[margin=1in]{geometry}

\usepackage{url}
\usepackage{xspace}
\usepackage{thm-restate}

\begin{document}
\title{Super-Polynomial Quantum Speed-ups for Boolean Evaluation Trees with Hidden Structure}
\author{Bohua Zhan\thanks{Department of Mathematics, Princeton University. Part of the work
conducted while at the Massachusetts Institute of Technology. Email: bzhan@princeton.edu} \text{  }  
Shelby Kimmel\thanks{Center for Theoretical Physics, Massachusetts Institute of Technology,
Email: skimmel@mit.edu. Supported by NSF Grant No. DGE-0801525,
{\em IGERT: Interdisciplinary Quantum Information Science and Engineering}.} \text{ }
Avinatan Hassidim\thanks{Google. Part of the work conducted while at Massachusetts Institute of Technology.
Email: avinatanh@gmail.com}  \text{  }}


\maketitle
\begin{abstract}
We give a quantum algorithm for evaluating a class of boolean formulas (such
as NAND trees and 3-majority trees) on a restricted set of inputs. Due to
the structure of the allowed inputs, our algorithm can evaluate a depth $n$ tree
using $O(n^{2+\log\omega})$ queries, where $\omega$ is independent of $n$ and depends only
on the type of subformulas within the tree. We also
prove a classical lower bound of $n^{\Omega(\log\log n)}$ queries, thus showing
a (small) super-polynomial speed-up.
\end{abstract}

\thispagestyle{fancy}
\rhead{MIT-CTP 4242}
\renewcommand{\headrulewidth}{0pt}
\renewcommand{\footrulewidth}{0pt}



\section{Introduction}
\label{intro}

In problems where quantum super-polynomial speedups are found, there is
usually a promise on the input.
Examples of such promise problems include
Simon's algorithm \cite{Simon1994}, the Deutsch-Jozsa algorithm
\cite{Deutsch1992} and the hidden subgroup problem \cite{Jozsa2001} (of which Shor's factoring algorithm
is a special case \cite{Shor1994}).

In fact, Beals et. al \cite{Beals2001} show that for a total boolean function (that is,
without a promise on the inputs), it is impossible to obtain a super-polynomial speedup over deterministic,
(and hence also over probabilistic),
classical algorithms. Expanding on the result, Aaronson and Ambainis showed in
\cite{Aaronson2009} that no super-polynomial speedup over probabilistic classical algorithms
is possible for computing a property of a function $f$ that is invariant under a permutation of the
inputs and the outputs,
even with a restriction on the set of functions $f$ considered.

In this paper, we take a quantum algorithm for a total boolean formula, which by \cite{Beals2001} cannot attain a
 super-polynomial speed-up, 
and show how to restrict the inputs to a point that a super-polynomial speed-up over probabilistic classical
algorithms can be attained. While super-polynomial speed-ups have been attained for boolean formulas (for example,
standard problems such as Deutsch-Josza could be written as a boolean formula) we don't know of other examples
that restrict the inputs to an extant problem, and which can be written naturally as a straightforward composition
of simple boolean gates.

 The total
boolean formulas we consider are boolean evaluation trees, such as the NAND tree, which has a quantum algorithm due to Farhi et. al.
\cite{Farhi2007}.
We show that
 the existing quantum algorithm of
Reichardt and $\check{\rm{S}}$palek \cite{Reichardt2009} for total boolean evaluation trees (with a small tweak) achieves a super-polynomial speed-up on our restricted set of inputs.
We choose our allowed set of inputs by closely examining the existing quantum algorithm for total functions
to find inputs that are ``easy'' for the algorithm.




While the restrictions we make on the domain are natural for a quantum algorithm, they are not so natural for a classical algorithm,
making our bound on classical query complexity the most technical part of this paper. We consider an
even more limited restriction on the inputs for the classical case, and show even with the added promise,
any probabilistic classical algorithm fails with high probability when less than
a super-polynomial number of queries are used. The additional restriction
considered in our classical proof leads to a problem similar
to a
problem considered by Bernstein and Vazirani called
Recursive Fourier Sampling (RFS, also known as the Bernstein-Vazirani Problem) \cite{Bernstein1993}. RFS is another example of a problem that acheives a super-polynomial speed-up. Extensions and lower bounds to RFS have been considered in \cite{Aaronson2003,Hallgren2008,Johnson2011}. We will describe the connections and differences between our problem and RFS later in this section.

Our problem is to consider a restricted set of inputs to a boolean evaluation tree.
An evaluation tree for a boolean function $f:\{0,1\}^c \rightarrow \{0,1\}$ is a complete
$c$-ary tree $T$ of depth $n$ where every node of $T$ is assigned a bit value. In general,
the leaves can have arbitrary values, and for every non-leaf node $v$ we have
\[Val(v) = f(Val(v_1), \ldots, Val(v_c)).\]
Here $Val(v)$ is the bit value of $v$, and $v_1, \ldots v_c$ are the children of $v$ in $T$.
We also sometimes say that a node $v$ corresponds to the function $f$ evaluated
at that node.
The value of the root is the value of the tree $T$. We want to determine
the value of $T$ by querying leaves of $T$, while making as few queries as possible.

For most functions $f$, the restriction on inputs is a bit artificial. For the NAND function, the
restriction can be given a slightly more natural interpretation. We will sketch this interpretation here
to give some intuition.
It is known that NAND trees correspond to game trees, with the value at each node denoting
whether the player moving can win under perfect play.\footnote{For more on this correspondence,
see Scott Aaronson's blog, \emph{Shtetl-Optimized},
 ``NAND now for something completely different," http://www.scottaaronson .com/blog/?p=207.}
Each turn, a player moves to one of two nodes,
and if both nodes have the same value, neither move will
change the game's outcome. However, if the two nodes have
different values, the choice is critical; one direction
will guarantee a win and the other a loss (under perfect play).
Then our restriction in the case of the game tree
is that for any perfect line of play (starting from
any node in the tree), the number of such decision
points is limited.
Farhi et al.
\cite{Farhi2007} showed that total NAND trees can be evalutated in $O(\sqrt{N})$ quantum queries, where $N$
is the number of leaves in the tree (corresponding to the number of possible paths of play in the game tree model). This is a polynomial speedup over the best
classical algorithm, which requires $\Omega(N^{.753})$ queries \cite{Saks1986}. One would expect that a tree
with few decision points would be easy to evaluate both quantumly and classically, and we will show that this
is indeed the case.

For the NAND/game tree, suppose we only allow inputs where on every path from root to leaf, there is
exactly one decision point, and on every path, they always occur after the same number of steps.
Then the oracle for this game tree is also a valid oracle for the 1-level RFS problem (with some
further restrictions on whether at the decision point moving right or left will cause you to win). Recall
that in the 1-level RSF, you are given a function $f:\{0,1\}^n\rightarrow\{0,1\}$ such that $f(x)=x\cdot s$
where $s$ is a hidden string. The goal of the problem is to output $g(s)$, where $g:\{0,1\}^n\rightarrow\{0,1\}$
is another oracle you are given access to. For general RFS, this 1-level problem is composed, and you are given different oracles $g_c$ for
each instance of the basic problem.
If we solve the NAND tree with one decision point on each path,
all at the same level, then we are essentially solving the 1-level RFS problem, but instead of plugging $s$ into
an oracle, the output of the NAND tree directly gives PARITY($s_n$), where $s_n$ is the last
bit of $s$.

While our problem is related to RFS, there are some significant differences. First, the analogous basic problem
in our case is most naturally written as a sequence of boolean gates, with the restriction on
inputs formulated in terms of these gates. Second, in our problem, the
output of each level is a bit that is then applied directly to the next level of recursion, rather
than using the unnatural convention of oracles for each level of recursion.
This lack of oracles makes proving the classical lower bound harder, as a classical algorithm could potentially
use partial information on an internal node.
Finally, this composed structure is only apparent in the classical bound; our quantum
algorithm applies to inputs such that the problem can't be decomposed into discrete layers, as in RFS.

We will in general consider evaluation trees made up of what we call {\it{direct}} boolean functions. We will
define such functions in Sec. \ref{SPsec}, but one major subclass of direct boolean functions is
threshold functions and their negations. We say that $f$ is a threshold function if there exists $h$ such that $f$ outputs 1 if
and only if at least $h$ of its inputs are 1. So the NAND function is a negation of the threshold function with $c=h=2$.
A commonly considered threshold function is the 3-majority function (3-MAJ) with $c=3$ and $h=2$.

We will now describe our allowed inputs to the evaluation tree.
 We first classify the nodes of the tree as follows:
for threshold (and negation of threshold) functions, each non-leaf node in the tree is classified as {\it{trivial}} if its children are
either all 0's or all 1's, and as a {\it{fault}} if otherwise.\footnote{In Section \ref{SPsec}, we will define trivial and fault nodes for non-threshold direct functions, and also see that it is possible
for other inputs to be trivial as well, but the all 0's and all 1's inputs are always trivial.}
So from our NAND/game tree example, decision points are faults.
Note trivial
nodes are easy to evaluate classically, since evaluating one child gives the value at the
node if the node is known to be trivial. It turns out that they are also easy to evaluate quantumly.
We next
classify each child node of each non-leaf node as either {\it{strong}} or {\it{weak}}.
If the output of a threshold function is 1, then the strong child nodes are those with value 1,
otherwise they are those with value 0. If a node is trivial then all children are strong.
Classically, the strong child nodes alone determine the value at the node,
and we will see that they matter more in computing the cost of the quantum algorithm as well.

Our promise is then that the leaves have values such that the tree satisfies the {\em $k$-faults} condition:

\begin{definition}\emph{($k$-fault Tree)}\label{defk}
Consider a $c$-ary tree of depth $n$, where throughout the tree nodes have been
designated as trivial or fault, and the child nodes of each node have been
designated as strong or weak in relation to their parent. All nodes (except leaves)
have at least 1 strong child node. For each node $d$ in the tree,
let $G_d$ be the set of strong child nodes of $d$. Then to each node $d$, we assign an
integer $\kappa(d)$ such that:
\begin{itemize}
\item $\kappa(d)=0$ for leaf nodes.
\item $\kappa(d)=\max_{b\in G_d} \kappa(b)$ if $d$ is trivial.
\item Otherwise $\kappa(d)=1+\max_{b\in G_d} \kappa(b)$.
\end{itemize}
A tree satisfies the $k$-faults condition if $\kappa(d)\le k$ for all nodes $d$ in the tree.
\end{definition}
\noindent In particular, any tree such that any path from the root to a leaf encounters only $k$ fault nodes
is a $k$-fault tree.


The main theorem we prove is

\begin{restatable}{theorem}{mainthem}\emph{(Main Result)} \label{main}
Given a $k$-fault depth $n$ tree with each node evaluating a fixed direct boolean function
$f_D$, we can create an algorithm based on a span program $P_D$ corresponding to $f_D$ that evaluates
the root of the tree with $O(n^2\omega^k)$ queries, for some constant $\omega$ which depends only on $f_D$.
\end{restatable}
A formula for $\omega$ will be given in Section \ref{SPsec}.


To prove the most generalized version of Theorem \ref{main}, we use the span-program
based quantum algorithm
of Reichardt and $\check{\rm{S}}$palek \cite{Reichardt2009}.
However, in Appendix \ref{NANDsec} we show that the original algorithm of Farhi et al.
\cite{Farhi2007} obtains a similar speedup, with the correct choice of parameters.


We use the quantum algorithm
of Reichardt and $\check{\rm{S}}$palek \cite{Reichardt2009}, which they only apply to total boolean evaluation trees (for a polynomial speed-up).
With our promise on the inputs, a tweak to their algorithm gives a super-polynomial speed-up. Their algorithm uses phase estimation of a quantum walk on weighted
graphs to evaluate boolean formulas. This formulation requires choosing a graph gadget to represent
the boolean function,
and then composing the gadget many times. The optimal gadget given the promise on the inputs may be different
from the one used in their original paper. While the original gadget is
chosen to optimize the worst case performance on any input, we should choose a possibly different gadget which is
very efficient on trivial nodes. This may increase
the complexity at faults, but this increase is bounded by the $k$-faults condition.

We also present a corresponding classical lower bound:
\begin{theorem}\emph{(Classical Lower Bound)}
\label{classthem}
Let $B$ be a (classical) randomized algorithm which finds the value of a depth $n$ tree composed
of direct boolean functions, satisfying the  $k$ faults condition for some $k < \text{polylog}(n)$.
If $B$ is correct at least 2/3 of the time on any distribution of inputs,
then the expected number of queries it makes is at least
$$\Omega \left((\log(n/k))^k \right)= \Omega \left( 2^{k\log \log (n/k) }\right).$$
\end{theorem}
When $k= \log n$ the number of queries is $n^{\Omega(\log \log n)}$, whereas the quantum
running time is $O(n^2\omega^{\log n}) = O(n^{2+\log\omega})$, which gives the speedup.

 In Section
\ref{SPsec} we will describe our quantum algorithm for $k$-fault direct boolean evaluation trees, based on the algorithm of Reichardt and $\check{\rm{S}}$palek \cite{Reichardt2009}, and prove Theorem \ref{main}. In Section
\ref{Classsec} we will sketch the proof of Theorem \ref{classthem} for the case of the
NAND tree. Full details of the lower bound proof can be found in the appendix.


\section{Quantum Algorithm} \label{SPsec}

\subsection{Span-Program Algorithm} Our algorithm is based on the formulation of \cite{Reichardt2009}, which uses span programs.
In this subsection, we will define span programs and the witness size,
a function that gives the query complexity of an algorithm derived from a given span program.

Span programs are linear algebraic ways of representing a boolean function.
We will define direct boolean functions, which are the functions
we use in our evaluation trees, based on their
span program representations.
(For a more general definition of span programs, see Definition 2.1 in \cite{Reichardt2009}).

\begin{definition}\emph{(Direct Boolean Function, defined in terms of its span program representation, adapted from Definition 2.1 in \cite{Reichardt2009})} \label{dirdef}
Let a span program $P_D$, representing a function $f_{D}(\vec{x})$, $\vec{x}=(x_1,\dots,x_c)$,
$x_j\in\{0,1\}$, consist
of a ``target" vector $t$ and ``input'' vectors $v_j: j\in\{1,\dots,c\},$ all of which are in
$\mathbb{C}^C$ $(C\in\mathbb{N}).$ Without loss of generality, we will always transform the program so that
$t=(1,0,$\dots$,0)$. Each $v_j$ is labeled by $\chi_j$, where $\chi_j$ is either $x_j$ or $\bar{x}_j$
depending on the specific function $f_D$ (but not depending on the specific input $\vec{x}$). The vectors $v_j$ satisfy the condition that $f_D(\vec{x})=1$ (i.e. true) if
and only if there exists a linear combination $\sum_ja_jv_j=t$ such that $a_j=0$
if $\chi_j$ is $0$ (i.e. false). We call $A$ the matrix whose columns are the
$v_j$'s of $P_D$: $A=(v_1,\dots, v_c)$. Any function that can be represented by such
a span program is called a direct boolean function.
\end{definition}

Compared to Definition 2.1 in \cite{Reichardt2009}, we have the condition that each input $x_j$ corresponds to
exactly one input vector - this ``direct" correspondance gives the functions their name.
As a result, for each direct boolean function there exists two special inputs,
$\vec{x}^{0}$ and $\vec{x}^1$, such that $\vec{x}^0$ causes all $\chi_j$ to be 0, and $\vec{x}^1$ causes all $\chi_j$ to be $1$. Note this means $\vec{x}^0$ and $\vec{x}^1$ differ at every bit, $f(\vec{x}^0)=0,$ $f(\vec{x}^1)=1$, and $f$ is monotonic
on every shortest path between the two inputs (that is, all paths of length $c$).



Threshold functions with threshold $h$
correspond to direct span programs where $\chi_j=x_j$ for
all $j\in[n]$, and where for any set of $h$ input vectors, but no set of $h-1$ input vectors,
there exists a linear combination that equals $t$. It is not hard
to show that such vectors exist. For threshold functions, $\vec{x}^0=(0,\dots,0)$ and
$\vec{x}^1=(1,\dots,1)$.


Any span program can lead to a quantum algorithm. For details, see Section 4 and Appendix B in
\cite{Reichardt2009}. The general idea is that a span program for a function $f$, with an input $\vec{x}$,
gives an adjacency matrix for
a graph gadget. When functions are composed, one can connect the gadgets to form larger graphs representing
the composed functions. These graphs have zero eigenvalue support on certain nodes only
if $f(x)=1$. By running phase
estimation on the unitary operator for a quantum walk on the graph, one can determine the value of the function with high probability.



Determining the query complexity of this
algorithm depends on the witness size of the span program:



\begin{definition} \emph{(Witness Size, based on Definition 3.6 in \cite{Reichardt2009})}
\label{wsize}
Given a direct boolean function $f_D$, corresponding span program $P_D$, inputs $\vec{x}=(x_1,\dots,x_c)$, $x_j\in\{0,1\}$,
and a vector $S$ of costs, $S\in[0,\infty)^c$, $S=(s_1,\dots,s_c)$,  let
$r_i\in\mathbb{C}^c$, $i\in\{0,\dots,C-1\}$ be the rows of $A$ and $\chi_j$ correspond to the columns
of $A$ (as in Definition \ref{dirdef}). Then the witness size is defined as follows:
\begin{itemize}
\item If $f_D(\vec{x})=1$, let $\vec{w}$ be a vector in $\mathbb{C}^c$ with components
$w_j$ satisfying $\vec{w}^\dagger r_{0}=1$, $\vec{w}^\dagger r_i=0$ for $i>1$, and
$w_j=0$ if $\chi_j=0.$ Then
\begin{equation}
\mathrm{wsize}_S(P_{D},\vec{x})=
\mathrm{min}_{\vec{w}}\sum_j s_j|w_j|^2.
\end{equation}
\item If $f_D(\vec{x})=0$, let $\vec{w}$ be a vector that is a linear combination of $r_i$,
with the coefficient of $r_0$=1, and with $w_j=0$ if $\chi_j=1$. Then
\begin{equation}
\mathrm{wsize}_S(P_{D},\vec{x})=
\mathrm{min}_{\vec{w}}\sum_j s_j|w_j|^2.
\end{equation}
\end{itemize}
\end{definition}
\begin{claim}
This definition is equivalent to the definition of witness size given in Definition 3.6 in \cite{Reichardt2009}. The reader can verify that we've replaced the dependence of the witness size on $A$ with a more explicit dependence on the rows and columns of $A$. For the case of $f(\vec{x})=0$ we use what they call
$A^\dagger w$ as the witness instead of $w$.
\end{claim}

{\bf{Notation:}} If $S=(1,\dots, 1)$, we leave off the subscript $S$ and write $\text{wsize}(P_D,\vec{x})$.

Now we will introduce a quantity called the subformula complexity. To simplify this paper, we will not go into the precise definition, which is not important for us, but rather focus on the relations between this quantity,
the witness size of a function, and the query complexity of that function. (If you would like to know more
about the subformula complexity, see Definitions 3.1 and 3.2 in \cite{Reichardt2009}. The following is adapted from Section 3 and 4.3 in \cite{Reichardt2009}).


Suppose we have a complete boolean evaluation tree.
We choose $|E|$ such that $1/|E|$ is larger
than the query complexity of the full tree. (Note, that since
the query complexity is normally somewhat large,
$|E|$ is a fairly small quantity.) We choose a definite value for $|E|$ later,
after partially calculating the
query complexity.
Consider a subtree within the larger tree
that evaluates the function $h(\vec{x})$, where $\vec{x}$ are literal inputs, (i.e.
the inputs are not the outputs of other functions). If $h$ has
span program $P_h$, then the subformula complexity $z$ of $h(\vec{x})$ is related to the
witness size by
\begin{equation}
z\leq c_1+\text{wsize}(P_h, \vec{x})(1+c_2|E|),
\end{equation}
where $c_1$ and $c_2$ are constants. The function $h$ is rooted at some node $v$. We will often call the
subformula complexity of $h$ the complexity of the node $v$. If $v$ is a leaf (i.e. a literal input), it has subformula complexity $z=1$.

Now we want to consider the subformula complexity of a composed formula:
$f=g(h(\vec{x}^1),\dots,h(\vec{x}^c))$. Let $P_g$ be a span program for $g$.
Let $z_j$ be the subformula complexity of $h(\vec{x}^j)$, and
$Z=(z_1,\dots,z_c)$.
Then the subformula complexity of $f$ is bounded by
\begin{eqnarray}
\label{wsizecomp}
z&\leq& c_1+\text{wsize}_Z(P_g,\vec{x})(1+c_2|E|\max_jz_j) \notag \\
&\leq &c_1+\text{wsize}(P_g,\vec{x})\max_jz_j(1+c_2|E|\max_jz_j).
\end{eqnarray}




Using this iterative formula, we can upper bound the subformula
complexity at any node in our formula.
We now choose $E$ such that
$E\ll z_v$ for all $v$, where $z_v$ is the subformula complexity at node $v$.
Then there exists (Section 4.4 from \cite{Reichardt2009}) a quantum algorithm to evaluate
the entire tree using $O(1/E)$ queries to the {\it{phase-flip input oracle}}
\begin{equation}
O_a:|b,i\rangle\rightarrow(-1)^{b\cdot a_i}|i\rangle,
\end{equation}
where $a_i$ is the value assigned to the $i^{th}$ leaf of the tree \cite{Reichardt2009}.

\subsection{$k$-fault Trees}\label{kfaultt}

In this section, we will create a span program-based quantum
algorithm for a $k$-fault tree composed of a single direct boolean function
with span program $P_D$, which requires
$O(n^2\omega^k$) queries, where $\omega=\mathrm{max}_{\vec{x}}\mathrm{wsize}(P_{D},\vec{x})$.
First, we will be more precise in our definitions of
{\it{trivial, fault, strong,}} and {\it{weak}}. Suppose we have a tree $T$
composed of the direct boolean function $f_D$, represented by the span program
$P_D$. Then,

\begin{definition} \emph{(Trivial and Fault)} A node in $T$ is trivial if it has input $\vec{x}$ where
$\mathrm{wsize} (P_{D},\vec{x})=1.$ A node in $T$ is a fault if it has input $\vec{x}$ where
$\mathrm{wsize} (P_{D},\vec{x})>1.$
\end{definition}
\noindent In calculating the query complexity of a boolean evaluation tree,
we multiply the witness sizes of the individual functions. Thus, to first order, any
node with $\text{wsize}=1$ doesn't contribute to the query complexity, and so is trivial.
We will show later (Thm. \ref{triv}) that for
direct boolean functions, we can always create
a span program that is trivial for inputs $\vec{x}^0$ and $\vec{x}^1$.

\begin{definition} \emph{(Strong and Weak)}
Let a gate in $T$ have inputs $\vec{x}=(x_1,\dots,x_c)$, and input labels $(\chi_1,\dots,\chi_c)$. Then
the $j^{th}$ input is strong if $f_D(\vec{x})=\chi_j$,
and weak otherwise.
\end{definition}

It should be clear that this definition of strong and weak agrees with the one for threshold
functions given in the introduction.


\begin{claim}
The costs $s_j$ corresponding to weak inputs do not affect the
witness size.
\end{claim}

\begin{proof}
From the definition of witness size (Definition \ref{wsize}), for $f_D(\vec{x})$ equals $0$ or $1$,
 we require $w_j=0$ for all $j$ where $\chi_j\neq f_D(x)$. This
means $s_j$ is multiplied by 0 for all such $j$, and therefore does not effect
the witness size.
\end{proof}

Now we can restate and prove Theorem \ref{main}:
\mainthem*
\begin{proof}
The basic idea is that to the first approximation, the query complexity can be
calculated by considering one path from root to leaf, taking the product
of all of the $\text{wsize}(P_D,\vec{x})$ that that path hits, and then
taking the maximum over all paths. (This can be seen from the second
line of Eq. \ref{wsizecomp}).  The condition on the maximum number of faults along
each path then gives a bound on the complexity throughout the tree, corresponding
to the factor of $\omega^k$ in the query complexity. It remains to
take into consideration corrections coming from $c_1$ and $c_2$ in Eq. \ref{wsizecomp}.

With our insight into strong and weak inputs, we can rewrite Eq. \ref{wsizecomp} as
\begin{align}\label{z0form} z\leq c_1+\mathrm{wsize}(P_{D},\vec{x})
(\mathrm{max}_{\mathrm{strong}\text{ }j}z_j)(1+c_2|E|\mathrm{max}_jz_j).
\end{align}
Note this applies to every node in the tree. Let the maximum energy $|E|$ equal $cn^{-2}\omega^{-k}$, where
$c$ is a constant to be determined. We will show that with this value,
the term $c_2|E|\mathrm{max}_jz_j$ will always be small, which allows us to explicitly calculate
the query complexity.

We will prove by induction that
\begin{equation}
z<c'(\bar{n}\omega^\kappa)(1+c_2cc'/n)^{\bar{n}}
\end{equation}
 for
each subtree rooted at a node at height $\bar{n}$,
where from Def. \ref{defk} $\kappa$ is an integer
assigned to each node based on the values of $\kappa$ at its child nodes and whether those nodes
are strong or weak. Here $c'$ is a constant larger than 1, dependent on $c_1$, and $c$ is chosen
such that $c_2cc'\ll 1$ and $c'c\ll 1$.

For leaves, $z=1<c'$, so the above inequality holds in the base case.
Consider a node with $\kappa=\eta$ and height $\bar{n}$. Then all input nodes have height
$\bar{n}-1$. Weak inputs have $\kappa \le k$ (notice the weak input subformula
complexities only show up in the last term $\mathrm{max}_jz_j$). If the
node is trivial then strong inputs have $\kappa \le \eta$, and if the node is a fault
then strong inputs have $\kappa \le \eta-1$. Assuming the appropriate values of $z_j$ based
on our inductive assumptions, then for the case of a trivial node,
Eq. \ref{z0form} gives
\begin{align} \label{trivcase}
z&\le c_1+c'(\bar{n}-1)\omega^\eta(1+c_2cc'/n)^{\bar{n}-1}(1+c_2cc'/n) \nonumber \\
&<c'\bar{n}\omega^\eta(1+c_2cc'/n)^{\bar{n}}.
\end{align}
Here we see that $c'$ is chosen large enough to be able to subsume the $c_1$ term into
the second term. For fault nodes, the bound on the complexity of inputs has an extra factor of $\omega^{-1}$
compared with the trivial case in Eq. \ref{trivcase},
which cancels the extra factor of $\omega$ from wsize, so the induction step holds in
that case as well.


For a $k$-fault, depth $n$ tree, we obtain $z_O<c'(n\omega^k)(1+c_2cc'/n)^n$ for all
nodes in the tree. Notice since $|E|=cn^{-2}\omega^{-k}$, $|E|\ll z$ for all nodes.
Based on the discussion following Eq. \ref{wsizecomp}, this means that
 the number of queries required by the algorithm is of order $1/|E|$, which is $O(n^2\omega^k)$.
\end{proof}

\subsection{Direct Boolean Functions}
The reader might have noticed that
our quantum algorithm does not depend on the boolean function being a
direct boolean function, and in fact, the algorithm applies to any boolean function
since any boolean function can be represented by a span program for which
at least one of the inputs is trivial. However, if the span program
is trivial for only one input, then it it is impossible to limit the number
of faults in the tree. Thus to fully realize the power of this promise,
we must have a boolean function $f$ with a span program that is trivial
on at least two inputs, $\vec{x}$ and $\vec{y}$
such that $f(\vec{x})=0$ and $f(\vec{y})=1$. This requirement
is also necessary to make the problem hard classically. In this section
we will now show that direct boolean functions satisfy this condition.

\begin{theorem} \label{triv}
Let $f_D$ be a direct boolean function, and let $P_D$ be a span program
representing $f_D.$ Then we can create a new span program $P_D'$ that also
represents $f_D$, but which is trivial on the inputs $\vec{x}^0$ and $\vec{x}^1$.
This gives us two trivial nodes, one with output 1, and the other with output 0.
\end{theorem}

\begin{proof}

If $P_D$ is our original span program with rows $r_i$ as in Definition \ref{wsize},
then we make a new span program $P_D'$ by changing $r_0$ to be orthogonal
to all other $r_i$'s for $i\geq1$. To do this, we let $R_1$ be the subspace spanned by the rows
$r_i$, $i\geq1$. Let $\Pi_{R_1}$ be the projector onto that subspace. Then we
take $r_0\rightarrow (\mathbb{I}-\Pi_{R_1})r_0$. Now $r_0$ satisfies $ r_i^\dagger r_0=0$ for $i>0$.
Looking at how we choose $\vec{w}$ in Def. \ref{wsize}, this transformation does not affect our choice of $\vec{w}$.
(For $f_D(\vec{x})=1$, $\vec{w}$ has zero inner product with elements of the subspace spanned by the
rows $r_i$, $i\geq1$, so taking those parts out of $r_0$ preserves $\vec{w}^\dagger r_0=1$,
which is the main constraint on $\vec{w}$. For $f_D(\vec{x})=0$, $\vec{w}$ is a sum of $r_i$ with $r_0$
having coefficient 1. But the part of $r_0$ that is not in the subspace $R_1$ will still have coefficient $1$, and we are free to choose the
coefficients of the $r_i$ ($i>0$) terms to make up the rest of $\vec{w}$ so that it is the same as before.)
Hence there is no effect on the witness size or the function represented by the span program.

We can now divide the vector space $R$ of dimension $c$ into three orthogonal subspaces:
\begin{equation}
R=(r_0)\oplus R_1\oplus R_2,
\end{equation}
where $(r_0)$ is the subspace spanned by $r_0$, $R_1$ is the subspace spanned
by $r_i$ for $i>0$, and $R_2$ is the remaining part of $R$. We can also write
part of the conditions for $\vec{w}$ in Claim \ref{wsize} in terms of these subspaces: if
$f_D(\vec{x})=1$, then $\vec{w}\in (r_0)\oplus R_2$, and if $f_D(\vec{x})=0$ then $\vec{w}\in (r_0)\oplus R_1.$

In both of these cases, the remaining $\vec{w}$ with the minimum possible length is proportional
to $r_0$, with length $|r_0|$ and $1/|r_0|$. When $\chi_j=0$ for all $j$, or $\chi_j=1$ for all $j$, i.e.
for inputs $\vec{x}^0$ and $\vec{x}^1$, there are no further restrictions
on $\vec{w}$, so these are the actual witness sizes. This shows one of the witness sizes for trivial
inputs must be at least 1. Seting $|r_0|=1$ by multiplying $r_0$ by a scalar, we can obtain both
witness sizes equal to 1. This gives a span program for $f_D$ such that for $\vec{x}^0$ and $\vec{x}^1$, $\mathrm{wsize} (P_{D},\vec{x})=1$.
\end{proof}

We have shown that we can make a span program for $f_D$ with inputs
$\vec{x}^0$ and $\vec{x}^1$
trivial.  However, the final scaling step that sets $|r_0|=1$ may increase the witness
size for other inputs (faults) compared to the original span program.
 Since we are limiting the number of faults, this doesn't hurt our query complexity.


\section{Classical Lower Bound Sketch}\label{Classsec}

The classical lower bound is proven by induction. We here give an overview of the proof for the NAND function. A more formal version (suited for a wider class of functions) with full details
is in Appendix \ref{apen}.

Consider  a boolean evaluation tree composed of NAND functions, of depth $n$, with no faults. The root and leaves all have the same value if and only if $n$ is even. Now,  suppose that every node at height $i$ is faulty, and there are no other  faults in the tree. In this case, all nodes at height $i$ have  value $1$ and the root has value $1$ if and only if $n - i$ is even.  Since $n$ is known, to determine the value of the root, an  algorithm has to determine the parity of $i$.

Let $T_1$ be the distribution on trees of depth $n$, which first picks $i$, the height of the fault (called also the split of the tree), uniformly at random\footnote{For other functions the basic construction still chooses a split, but the split no longer consists of one level of faults. Instead there is a constant number of faults in each split. This does not change the asymptotic behavior. The appendix uses the more general construction, which does not map to this one for NAND functions (for example, if one were to apply the general construction to the  NAND function, the split could involve two levels).}, and then picks at random, for each node at height $i$, which one of its children has the value $0$, and which one has the value $1$. Consider  two leaves $u,v$ which the algorithm queries. If their common ancestor  has height less than $i$, $u$ and $v$ will have the same value. If the  height is more than $i+1$, the probability that they have the same value  is $0.5$, and if the height is exactly $i$, then the values will  differ. Thus, by querying different leaves, the algorithm performs a (noisy) binary search, which enables it to determine the parity of $i$.

It is easy to show that determining the parity of $i$ is not much easier than determining $i$. In fact, we show that for every  algorithm which gets a tree from the above distribution, if the  algorithm queries at most $\beta = \log n / 5$ leaves, the probability  that it can make a guess which has an advantage better than $n^{-1/5}$  over a random guess is at most $n^{-1/5}$.  Note that this type of bi-criteria guarantee is necessary: an algorithm could (for example) ask  the leftmost leaf, and then $\beta$ leaves which have common ancestors  at height $1, 2, \dots, \log n/5$ from that leaf. If the algorithm uses this  type of strategy, it has a small chance to know exactly in which height  the fault (the split) occurred: it will find it in all trees in which it occurs in  the first $\beta$ levels. On the other hand, the algorithm could  perform a regular binary search, and narrow down the range of possible  heights. This will give the algorithm a small advantage over a random  guess, regardless of the value of $i$.

Since we have algorithms with a super-polynomial number of queries, we cannot apply some union bound to say that  rare events (in which a dumb algorithm makes a lucky guess and learns a value with certainty) never happen. Thus, we must follow rare events carefully throughout the proof.

The distribution we use for trees with multiple faults is recursive. We define a distribution $T_k$ on trees of height $nk$ which obey the $k$ fault rule\footnote{As $n$ is much  larger than $k$, the assumption that the height of the tree is $n k$ and  not $n$ does not alter the complexity significantly - one can just  define $m = n/k$, and the complexity will be $m^{O(\log \log m)}$ which  is also $n^{O(\log \log n)}$ for $k = O(\log n)$.}. We begin by sampling $T_1$ (remember $T_1$ is the distribution of the trees of height $n$ which obey the $1$-fault rule, and which have faults exactly at every node at height $i$).  Then, we replace each one of the leaves of this depth $n$ tree with a  tree of depth $n(k-1)$, which is sampled from the distribution $T_{k-1}$, where we  require that  the sampled tree has the correct value at the root (the  value at the root of the tree of depth $n(k-1)$ needs to be identical to  the value of the corresponding leaf of the tree of depth $n$). Note  that the tree generated this way is identical to the one generated if we took a depth  $n(k-1)$ tree generated recursively (sampling $T_{k-1})$, and expanded each one  of its leaves to a tree of depth $n$ (sampling $T_1$). However, it is easier to  understand the recursive proof if we think of it as a tree of depth $n$,  in which each leaf is expanded to a tree of depth $n(k-1)$.

We now present the inductive claim, which is also a bi-criteria. Since we will apply the inductive assumption
to subtrees, we need to consider the possibility that we already have some a priori knowledge about the value
at the root of the tree we are currently considering. Moreover this prior can change with time. Suppose that at a given time, we have prior information which says the value of the root of a $T_k$ tree is $1$ with some probability $p$. We prove that if the algorithm has made less than $\beta^l$ queries on leaves of the $T_k$ tree, for $l\leq k$, then the  probability that it has an advantage greater than $c_{k,l} <O(2^{k-l+1}n^{-3(k-l+1)/5}\beta^{4+6(k-l)})$ over  guessing that the root is $1$ with probability $p$, is at most $p_{k} <O( 3n^{-1/5})$. Note that the probability that the algorithm guesses the root correctly varies as $p$ varies. We require that with high probability during all that process that the advantage the algorithm gets by using the information from leaves below the root of the $T_k$ tree is bounded by $c_{k,l}$.

The proof requires us to study in which cases the algorithm has an advantage greater than $c_{k,l}$ when guessing a node, although it asked less than $\beta^{l}$ questions. We say that an ``exception'' happened at a vertex $v$, if one of the  following happened:

\begin{enumerate}[I]
   \item There was an exception in at least two depth $n$ descendants.
  \item   ``The binary search succeeds:'' We interpret the queries made by  the  algorithm as a binary search on the value of the root, and it  succeeds  early.
   \item There was one exception at a depth $n$ descendant, when there was  some strong prior (far from 50/50) on the value of this child, and the  value discovered is different than the value which was more probable.
\end{enumerate}

Every type III Exception is also a type II, and a type I exception requries at least two type I, II, or III exceptions in subtrees. We show that if the algorithm has an advantage of more than $c_{k,l}$ at a certain node with less than $\beta^{l}$ queries, then an exception must have occured at that node.

The proof is by induction. First we show that the probability of an expection at a node of height $n k$ is bounded by $p_k$. The second part bounds $c_{k,l}$, by conditioning on the (probable) event that there was no exception at height $nk$.  Applying the inductive claim with $k$ equal to the maximum number of faults in the original
tree and no a priori knowledge of the root gives the theorem.

The first part of the proof is to bound the probability of exception, $p_k$. For type II exceptions this
bound comes from the probability of getting lucky in a binary search, and does not depend on the size of the tree. Since type I exceptions depend on exceptions on subtrees, intuitively we can argue that
since the probability of exception on each subtree is small from the $k-1$ case, the probability for a type I is also small.
A technical problem in making this rigorous is that the algorithm has time to evaluate a few leaves in
a large number of subtrees, then choose the ones that are more likely to give exceptions to evaluate
further. Then the inductive bound from the $k-1$ case no longer applies to the subtrees with the most
evaluations. The solution used in the paper is to simulate $A$, the original algorithm, by another
algorithm $B$, that combines the evaluations by $A$ into at most $\beta$ subtrees, and then argue that
the probability of exception in $B$ cannot be much lower than that in $A$. This way, the inductive
bound on $p_{k-1}$ can be applied directly.

For type III exceptions, we note that at each step where a type III exception is possible, the probability
of exception at the subtree is cut by a constant factor $C$. Intuitively, this means it
is $C$ times less likely to have a type III exception. Another simulation is used to
formulate this proof rigorously.

For  the second part, we consider a tree of height $n$ from our distribution rooted at a node $u$  (which need not be the root of the entire tree), where each one of its  leaves is a root of a tree of height $n(k-1)$. Let $L$ denote the set of  nodes at depth $n$, which are each roots of trees of height $n(k-1)$.  The algorithm is allowed $\beta^{l}$ queries, and thus there are at  most $\beta$ nodes in $L$ which are the ancestors of $\beta^{l- 1} $  queries, at most $\beta^2$ nodes which are the ancestors of $\beta^{l  - 2}$ queries etc. Suppose on each node $v \in L$ which received  between $\beta^{j-1}$ and $\beta^j$ queries, we know the information  $c_{k-1,j}$. Suppose also, that one of the nodes in depth $L$ is known,  although it was queried an insufficient number of times (since we condition on the fact that there is no exception at height $nk$, this assumption is valid). Adding the  information we could obtain from all of these
possible queries gives the bound on $c_{k,l}$.

\section{Conclusions}
We have shown a restriction on the inputs to a large class of boolean formulas results in a super-polynomial quantum speed-up. We
did this by examining existing quantum algorithms to find appropriate restrictions on the inputs.
Perhaps there are other promises (besides the $k$-faults promise) for other classes of functions (besides direct boolean functions) that could be found in this way. Since we are interested in understanding which restrictions lead to super-polynomial speed-ups, we hope to determine if our restricted set of inputs is the largest possible set
that gives a super-polynomial speed-up, or whether the set of allowed inputs could be expanded.

The algorithm given here is based on the fact that there is a relatively large gap around zero in the spectrum
of the adjacency matrix of the graph formed by span programs.
That is, the smallest magnitude of a nonzero eigenvalue is polynomial in $n$, or logarithmic in the size of
the graph. From the analysis of span programs of boolean functions
we now have a rather good idea about what kind of tree-like graphs have this property.
Furthermore Reichardt has found a way to collapse such graphs into few level graphs that have
this same spectrum property \cite{Reichardt2009a}, giving some examples with large cycles.
One problem is to determine more generally which graphs have this property about their spectrum.


\section{Acknowledgments}
We thank Eddie Farhi, Jeffery Goldstone, Sam Gutmann, and Peter Shor for many insightful discussions.


\bibliography{superpbib}
\bibliographystyle{abbrv}

\appendix
\section{Intuition from the NAND-Tree} \label{NANDsec}

For those unfamiliar with span programs, which form the basis of our general
algorithm, we will explain the key ideas in the setting of NAND
trees from Farhi et al. \cite{Farhi2007}. In their algorithm a
continuous time quantum walk is performed along a graph, where amplitudes on the nodes of the
graph evolve according to the Hamiltonian
\begin{equation} \label{Hamil}
H|n\rangle=\sum_{\mathrm{n_i: \text{ }neighbors\text{ }of\text{ }n}}-|n_i\rangle.
\end{equation}
The NAND function is encoded in a Y-shaped graph gadget, with amplitudes of the
quantum state at each node labelled by $u_1$, $u_2$, $w$ and $z$, as shown in Figure \ref{NANDY}.
\begin{figure}[h]
\center\includegraphics[scale=.4]{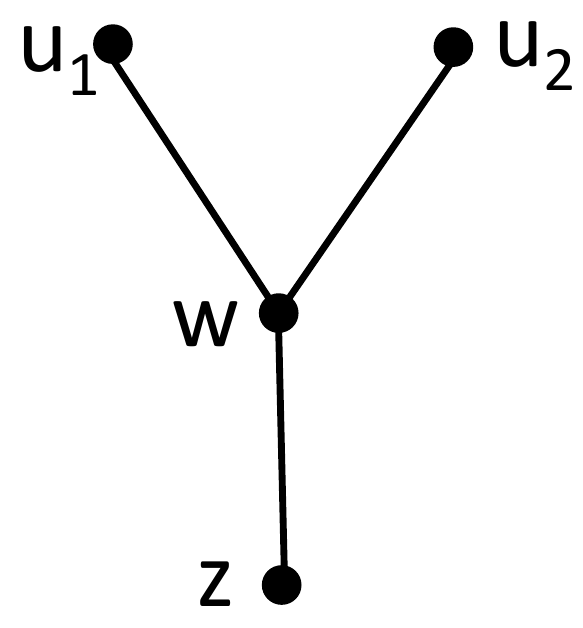}
 \caption{Graph gadget for the NAND function, with amplitudes labelled by $u_1$, $u_2$, $w$ and $z$.}
 \label{NANDY}
  \end{figure}

For an eigenstate of the graph with eigenvalue $E>0$, associate to each node the ratio of the
amplitude at the parent of that node to the amplitude at the node itself. For example,
in Fig. \ref{NANDY}, the ratios associated with the top-most nodes are called input ratios, and have values
$y_1(E)=\frac{u_1}{w}$ and $y_2(E)=\frac{u_2}{w}$.
Then using the Hamiltonian in Eq. \ref{Hamil}, we can find the output ratio (associated with the central node)
$y_{0}(E) =\frac{w}{z}$ in terms of the input ratios:
\begin{equation} \label{outputratio}
y_{0}(E)=-\frac{1}{y_1(E)+y_2(E)+E}.
\end{equation}
Farhi et al. associated the ratio at a node with the literal value at that node. They showed that
for eigenstates with small enough $E$, the literal value 1 at a node corresponds
to a ratio at most linear in $E$ that is, $0\le y_i(E)<a_iE$ for some $a_i$. The literal value of 0
corresponds to a negative ratio, which has absoulte value at least $O(1/E)$,  written $y_i(E)<-1/(b_iE)$ for some $b_i$. Using the recursive
relation (\ref{outputratio}), and the correspondance between literal values
and ratios, we will show below that
the Y gadget really does correspond to a NAND gate.

We require $E$ small enough so that $Ea_{i}$ and $Eb_{i}$
are small. We call $a_{i}$ and $b_{i}$ complexities, and each input ratio
or output ratio has a complexity associated with it. The maximum complexity seen
anywhere in the tree determines the maximum allowed eigenvalue $E_0$, which in turn controls the runtime
of the algorithm. We will determine the complexity at the output ratio of a Y gadget
(the output complexity), given the complexities of the input ratios (the input complexities).

\begin{itemize}
\item {\bf Input $\{00\}$}. For simplicity, we assume the input complexities are equal,
with $y_1(E)=y_2(E)=\frac{-1}{bE}.$ Applying Eq. \ref{outputratio}, $y_0(E)=bE/(2-bE^2)$.
 To first order, $y_0(E)=\frac{bE}{2}$. Thus, the output ratio has
 literal value 1, with complexity $\frac{b}{2}.$ The output complexity is half of
 the input complexity.

\item {\bf Input $\{11\}$}. For simplicity, we assume the input complexities are equal,
with $y_1(E)=y_2(E)=aE.$ So $y_0(E)=-1/(E+2aE).$
To first order, $y(E)=\frac{-1}{2aE}$. Thus, the output ratio has
literal value 0, with complexity $2a.$ The output complexity is double
the input complexity.

\item {\bf Input $\{10\}$}. Then $y_1(E)=aE$ and $y_2(E)=\frac{-1}{bE}.$
So $y(E)=bE/(1-(1+a)bE^2).$ To first order, $y(E)=bE$. Thus, the output ratio has
literal value 1, with complexity $b$.
The output complexity equals the input complexity of the $0$-valued input.
\end{itemize}

We will show how the terms introduced in Section \ref{intro} (fault, trivial, strong, and weak), apply
to this example. Recall a node with input $\{00\}$ or $\{11\}$ is trivial, with both inputs strong.
A node with input $\{01\}$ or $\{10\}$ is a fault, with the 0-valued input strong and the 1-valued input weak.

Generally, weak inputs matter less when calculating $\kappa$ from Section \ref{intro}, and we see in this
example that the input complexity of the weak input does not affect the output complexity (as long
as it is not too large compared to the eigenvalue). Next consider composing NAND functions by
associating the output ratio of one gadget with the input ratio of another gadget. When we only have
trivial nodes, the complexity is doubled and then halved as we move between nodes with input $\{11\}$
and those with input $\{00\}$, resulting in no overall multiplicative increase in the complexity. So trivial
nodes do not cause a multiplicative increase in the complexity and therefore do not cause a multiplicative
increase in the runtime.

Whenever there is a fault, its strong input must have come from a node with input $\{11\}$. This node
doubles the complexity, and then the fault itself does not increase the complexity. Thus
the presence of the fault guarantees an overall doubling of the complexity coming from the strong input.

Based on this analysis, the maximum complexity of a $k$-faults NAND tree is
$2^k$, since $k$, as calculated in Section \ref{intro}, corresponds in our example
to the maximum number of times the complexity is doubled over the course
of the tree. In turn, the complexity
corresponds to the runtime, so we expect a runtime of $O(2^k)$.
(That is, the constant $w$ is 2 for the NAND tree.)
This analysis has ignored a
term that grows polynomially with the depth of the tree, coming from higher order corrections. This extra term
 gives us a runtime of $O(n^22^k)$ for a depth $n$, $k$-fault NAND tree.
%
%
%


\section{Classical Lower Bound} \label{apen}

\subsection{Lower Bound Theorem for Direct Boolean Functions}

Consider a
complete $c$-ary tree composed of direct boolean functions $f_D(x_1,\dots,x_c).$
Each node $d$ in the tree $t$ is given a boolean value $v(t,d)\in\{0,1\};$ for
leaves this value is assigned, and for internal nodes $d$, it is the
value of $f_D(v(t,d_1)\dots,v(t,d_c))$ where $(d_1,\dots,d_c)$ are the child nodes of $d$. If
clear from context, we will omit $t$ and write $v(d)$.

We will prove the following theorem:
\begin{theorem} \label{mainthm}
Let $f_D:\{0,1\}^c\rightarrow\{0,1\}$ be a direct boolean function that
is not a constant function. Suppose there exist $k_0$ and $n_0$ such that for each $r\in\{0,1\}$,
there is a distribution $\mathcal{G}_r$ of trees composed of $f_D$ with height
$n_0$, root node $r$, and satisfying the $k_0$-faults condition, such that a
priori, all leaves are equally likely to be 0 or 1. Then given $n \gg n_0$ and $k$
polynomial in $\log n$, for a tree $t$ composed of $f_D$ with height $k\cdot n$
and satisfying the $(k\cdot k_0)$-faults condition, no probabilistic classical algorithm
can in general obtain $v(t,g)$ with at least confidence $2/3$ before evaluating
$\beta^k$ leaves, where $\beta=\lfloor \mathrm{log}\tilde n/10 \rfloor$, and
$\tilde{n}=n-n_0$. (Note $a\cdot b$ denotes standard multiplication.)
\end{theorem}

\noindent An example of $\mathcal{G}_r$ for the NAND tree can be seen in Figure \ref{G_r}.
\begin{figure}[h]
\center\includegraphics[scale=.4]{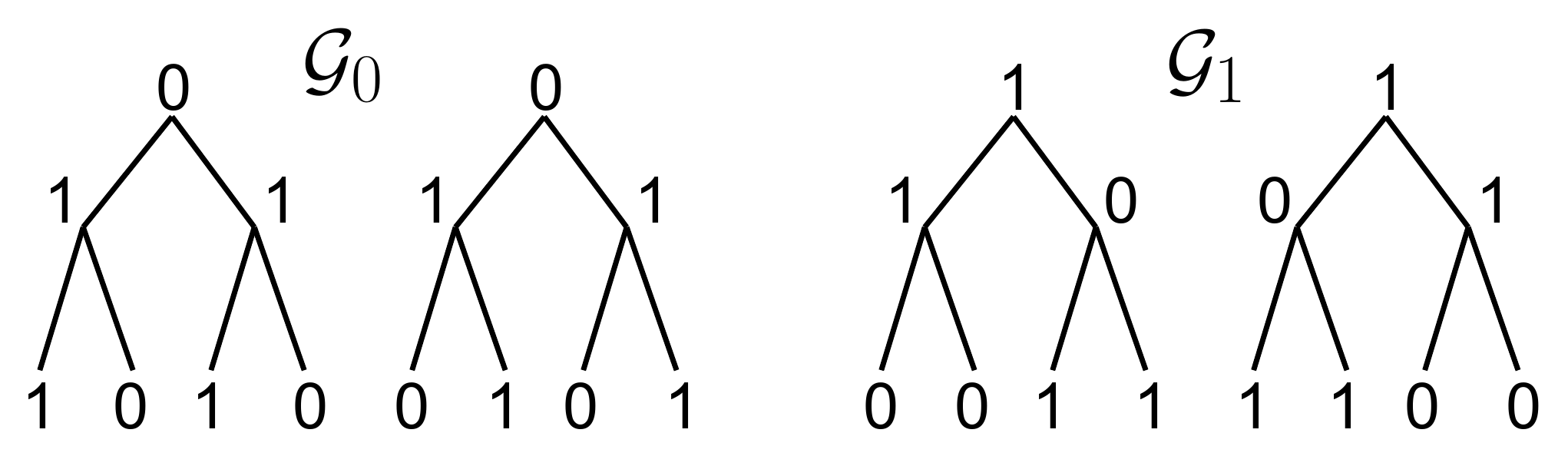}
 \caption{Examples of the distributions $\mathcal{G}r$ for the $\textsc{NAND}$ tree. Each $\mathcal{G}_r$ is composed
by choosing one of the two trees in its distribution with uniform probability. In this example, $k_0=1$, $n_0=2$.}
 \label{G_r}
  \end{figure}

For example, if $k=\lfloor \log n\rfloor$, then the query complexity is
$\Omega((\mathrm{log}n)^{\mathrm{log}n})=\Omega(n^{\mathrm{loglog}n})$, which gives us a
super-polynomial separation from the quantum algorithm. To prove Theorem \ref{mainthm}
we will construct
a `hard' distribution such that for a tree randomly drawn from this distribution, the best
deterministic algorithm (tailored to this hard distribution) cannot succeed with
high confidence. Then by Yao's minimax principle \cite{Yao1977}, the best
probabilistic algorithm can on average do no better.

\subsection{Construction of a Hard Distribution}
Choose $f_D$, $n_0$, $k_0$ according to Thm. \ref{mainthm}. Fix $n$ such that
$n\gg n_0$. For each $k\geq1$ we construct a distribution $\mathcal{T}_k$ of
trees with height $k\cdot n$ satisfying the $(k\cdot k_0)$-fault rule. For
a tree chosen at random from $T_k$ we show no deterministic algorithm can obtain
the value at the root with probability $2/3$ before evaluating $\beta^k$ of the
leaves, implying Thm \ref{mainthm}.

\begin{claim}
Given a function $f_D$ and a span program adjusted so that $\chi=\{0,1\}$ make
the function trivial, let $(x_{w1},\dots,x_{wc})$, for $w=\{0,1\}$, be the
literal inputs corresponding to $\chi=\{0,1\}$ respectively. Then given
$r\in\{0,1\}$, $n\geq1,$ there exists a unique tree $t_{r,n}$ with height $n$,
root node $g$, and $v(t_{r,n},g)=r$, such that for any non-leaf node $d$ in
the tree, the values of the child nodes of $d$ are $(x_{11},\dots,x_{1c})$ or $(x_{01},\dots,x_{0c})$. Moreover,
$t_{0,n}$ and $t_{1,n}$ differ at every node of the tree.
\end{claim}
\begin{proof}
Recall that $\chi_j=0$ for all $j$ when $\chi=0$, and $\chi_j=1$ for all $j$ when
$\chi=1$. Hence if $x_{0j}=x'_j$, then
$x_{1j}=\bar{x}'_j$ because $\chi_j$ is represented by $x_j$
or $\bar{x}_j$ depending on the span program, and independent of the input. (That is,
if $\chi_j$ always is the negation of $x_j$, then if $x_j=1$, $\chi_j=0$, and if $x_j=0$,
$\chi_j=1$.)
 So for a node $d$ in $t_{r,n}$,
if $d=1$ it must have inputs $(x_{11},\dots,x_{1c})=(x'_1,\dots,x'_c)$ and if $d=0$, it must
have inputs $(x_{01},\dots,x_{0c})=(\bar{x}'_1,\dots,\bar{x}'_c)$. As soon as the value of
the root is chosen, there is no further choice available in $t_{r,n}$, giving
a unique tree. Because $x_{0j}=\bar{x}_{1j}$, $t_{0,n}$ and $t_{1,n}$ differ at
every node.
\end{proof}

We construct $\mathcal{T}_k$ by induction on $k$. First we construct the distribution
$\mathcal{T}_1$. Randomly choose a root value $r\in\{0,1\}$ and a category value
$i\in\{1,\dots,\tilde{n}\}$. Then up to depth $i$ insert the
tree $t_{r,i}$. At each node $d$ at depth $i$, choose a tree according to the
distribution $\mathcal{G}_{v(t,d)}$ to be rooted at $d$. The trees in
$\mathcal{G}_r$ have height $n_0$, so we are now at depth $(i+n_0)$.
 At each node $d$ at depth $(i+n_0)$ attach the tree
$t_{v(t,d),\tilde{n}-i}$.
Since all of the nodes in $t_{r,n}$
are trivial, {\it{any}} path from the root to the leaves contains at most
$k_0$ faults, so $t$ satisfies the $k_0$-faults rule. We call $i$ the level
of the fault, since faults can only occur from depth $i$ to depth $i+n_0$.
Let $\mathcal{T}_{1,r,i}$
be a distribution of trees constructed as above, with root $r$ and category $i$.
Then $\mathcal{T}_1$ is a combination of the distributions $\mathcal{T}_{1,r,i}$ such
that each distribution $\mathcal{T}_{1,r,i}$ with $r\in\{0,1\}$,
$i\in\{1,\dots,\tilde{n}\}$ is equally likely.

For the induction step, we assume we have a distribution of trees
$\mathcal{T}_{k-1}$. Then choosing a tree from the distribution
$\mathcal{T}_k$ is equivalent to choosing a tree $t$ from $\mathcal{T}_1$
and then at each leaf $a$ of $t$, attaching a tree from $\mathcal{T}_{k-1}$
with root value $v(t,a)$. We call $\mathcal{T}_{k,r,i}$ the distribution of
trees in $\mathcal{T}_k$ with root $r$ and category $i$ in the
{\it{first}} $n$ levels of the tree. As in the $\mathcal{T}_1$ case,
$\mathcal{T}_k$ is a combination of the distributions $\mathcal{T}_{k,r,i}$
such that each distribution $\mathcal{T}_{k,r,i}$ with $r\in\{0,1\}$,
$i\in\{1,\dots,\tilde{n}\}$ is equally likely. We'll call $\mathcal{T}_{k,r}$
the distribution of trees in $\mathcal{T}_k$ with root $r$.

Before we state the main lemma needed to prove Thm \ref{mainthm}, we need a
few definitions.

\begin{definition}
Suppose we have evaluated some leaves on a tree drawn from
the distribution $\mathcal{T}_k$. Let $d$ be a node in a tree, such
that $d$ is also the root of a tree $t\in\mathcal{T}_j$. Then $p_{root}(d,r)$
is the probability that a randomly chosen tree in $\mathcal{T}_{j,r}$ is not
excluded as a subtree rooted at $d$, where the probability is determined solely
by the evaluated leaves in $t$, that is, only from leaves descending from $d$.
\end{definition}

For example, if $g$ is the root of the entire tree, then $p_{root}(g,r)$ is
the probability that a random tree from the distribution $\mathcal{T}_{k,r}$
is not excluded, based on all known leaves. If on the other hand, we look at
a node $a$ in $\mathcal{T}_k$ at depth $n$, then $p_{root}(a,r)$ is the probability that a
tree from $\mathcal{T}_{k-1,r}$ is not excluded as the tree attached at $a$,
where the decision is based only on evaluated leaves descending from $a$.

Let $g$ be the root node. Then our confidence in the value at the root of our
tree depends on the difference between $p_{root}(g,0)$ and $p_{root}(g,1)$.

\begin{definition} \label{confdef}
We define the confidence level as
\begin{align}
\frac{|p_{root}(g,0)-p_{root}(g,1)|}{p_{root}(g,0)+p_{root}(g,1)}.
\end{align}
\end{definition}
So our confidence level is 0 if we have no knowledge of the root based on
values of leaves descending from $g$, and our
confidence level is 1 if we are certain of the root's value based on
these values.

Our goal will be to bound the confidence level if less than a certain
number of leaves have been evaluated. In calculating the confidence
level, we use $p_{root}$ rather than the true probability that the root $g=0$
because we will be using a proof by induction. Thus we must think of $g$ as
 being the root of a subtree that is an internal node in a larger tree.
Evaluations on other parts of this larger tree could give us information about
$g$. In fact, if evaluations on the subtree rooted at $g$ are interwoven with
evaluations on the rest of the tree, our outside information about $g$ will
change between each evaluation we make on the $g$-subtree. We could even
learn the value of $g$ with near certainty from this outside information. This
outside information could then be used to help us decide which leaves descending
from $g$ to evaluate. However, we will show that the outside information
will not help us much to increase the confidence level considering only leaves descending from $g$.
\begin{definition}
The relation $p\lesssim q$ ($p$ asymptotically smaller than $q$), means
$p<q(1+C\beta^\alpha\tilde{n}^\gamma)$, for fixed constants
$\alpha>0$, $\gamma<0$, and $C$ growing polynomial in $k$. Similarly
$p\gtrsim q$ means $p\geq q(1-C\beta^\alpha\tilde{n}^\gamma)$.
\end{definition}

Since $k$ is polynomial in $\log n$ and $\beta =\lfloor \mathrm{log}\tilde n/10 \rfloor$
by Theorem \ref{mainthm},  $C\beta^\alpha$ is
polynomial in $\log n$, and for $n$ large, $C\beta^\alpha\tilde{n}^\gamma$
contributes only a small factor.

Our main lemma is:

\begin{lemma}\label{mainLemma}
For each $1\leq l\leq k$, suppose the algorithm made less than $\beta^l$
evaluations on a tree chosen from some weighted combination of $\mathcal{T}_{k,0}$
and $\mathcal{T}_{k,1}$, where the weighting may change due to outside knowledge
gained between each evaluation.
Then with probability at least $1-p_{k}$, the confidence
level as defined in Def. \ref{confdef} is less than $c_{k,l}.$ The values of
$p_{k}$ and $c_{k,l}$ are:
\begin{align}
p_{k}&\lesssim 3\tilde{n}^{-1/5}\notag\\
c_{k,l}&\lesssim(8n_0)^{k-l+1}\tilde{n}^{-3(k-l+1)/5}\beta^{4+6(k-l)}.
\end{align}
 \label{mainlem}
\end{lemma}
The part of the lemma with $l=k$ and an equally weighted combination of
$\mathcal{T}_{k,0}$ and $\mathcal{T}_{k,1}$ implies Thm \ref{mainthm}.

\section{Proof of Lemma 1}
We will prove Lemma \ref{mainLemma} in this section.

In general, the confidence level as defined in Def. \ref{confdef} is too broad
to tackle directly. Instead, we will rewrite the confidence level in terms of
another probability $p_{cat}$ (``cat'' is for category).

\begin{definition}
Let $d$ be a node in a tree drawn from the distribution $\mathcal{T}_k$, such
that $d$ is also the root of a tree $t\in\mathcal{T}_j$. Then $p_{cat}(d,i,r)$
is the probability that a randomly chosen tree in $\mathcal{T}_{j,r,i}$ is not
excluded as a subtree at $d$, where the probability is determined solely by the
known values of leaves descending from $d$. If $d$ is clear from context, then
we will simply write $p_{cat}(i,r)$
\end{definition}
For each node $d$ as above, let $S(d)=\sum_{i,r}p_{cat}(d,i,r)$ and
$D(d)=\sum_i|p_{cat}(d,i,0)-p_{cat}(d,i,1)|$. Again the parameter $d$ may be
omitted if it is clear from context. Suppose $g$ is the root of the tree. Then
the confidence level can be written as
\begin{align}
&\frac{|p_{root}(g,0)-p_{root}(g,1)|}{p_{root}(g,0)+p_{root}(g,1)} \notag \\
&=\frac{|\sum_i p_{cat}(g,i,0)-\sum_i p_{cat}(g,i,1)|}{\sum_{i,r} p_{cat}(g,i,r)}&\notag \\
&\leq\frac{\sum_i|p_{cat}(i,0)-p_{cat}(i,1)|}{\sum_{i,r}p_{cat}(i,r)}& \notag \\
&=\frac{D(g)}{S(g)}&
\end{align}
Our general approach in proving Lemma \ref{mainlem} will be to find an upper bound
on $D(g)$ and a lower bound on $S(g)$, thus bounding $D(g)/S(g)$ and hence the
confidence level. We will generally refer to $D(g)$ as $D$ and $S(g)$ as $S$.

\subsection{Basis Case: $k=1$}
The lower bound on $S$ comes from creating an analog with a binary search.
The following lemma will be used in the base case, and in the induction step:
\begin{lemma} \label{intlem}
We have an interval of length $A_0$ to be divided into smaller intervals. At
each step $\tau\in\mathbb{Z}^+$, we choose a number $p_\tau$ between $0$ and
$1$. If our current interval is $A_\tau$, then $A_{\tau+1}=p_\tau A_\tau$
with probability $p_\tau$, and $A_{\tau+1}=(1-p_\tau)A_\tau$ with probability
$1-p_\tau$. Then for $m\in\mathbb{Z}^+$ and $0<F<1$, the probability that
$A_m<FA_0$ is less than $2^mF$, regardless of the choices of $p_\tau$.
\end{lemma}
\begin{proof}
We model the process as a binary search with uniform prior. We try to guess a
real number $x$, $0<x<A_0$ where at each step we guess the value $y_\tau$ and
are told whether $x<y_\tau$ or $x>y_\tau$. Then $A_\tau$ is the size of the
interval not excluded after $\tau$ steps, and $p_\tau$ is related to the
position of $y_\tau$ within the remaining interval. For each strategy of
choosing a sequence of $p_\tau$'s, $\tau\in\{0,\dots,m-1\}$, where each $p_\tau$
can depend on all previous outcomes and choices of $p_{\tilde{\tau}}$, for
$\tilde{\tau}<\tau$, there are at most $2^m$ possible outcome subintervals $A_m$.
These subintervals fill the interval $A_0$ completely and without overlap.
Then the condition $A_m<FA_0$ corresponds to $x$ falling in an interval with
size less than $FA_0$. Since there are only $2^m$ such intervals, the combined
length of intervals with size less that $FA_0$ is less than $2^mFA_0$. Since
we started with an interval of length $A_0$, the probability of being in one
of these intervals is less than $2^mF$.
\end{proof}

Consider a tree $t$ chosen according to the distribution $\mathcal{T}_1$.
We'll call the root node $g$ (so throughout this section, whenever we write
$p_{cat}(i,r)$ we mean $p_{cat}(g,i,r)$). Suppose we have so far evaluated
$\tau-1$ leaves, and thus excluded a portion of the possible trees in
$\mathcal{T}_1$. We are about to evaluate the $\tau^{th}$ leaf, which we'll
call $b$. Let $S_\tau$ be the current value of $S$, and let $S_{\tau+1}$ be
the value of $S$ after $b$ is measured.

Let $p_\tau$ be the true probability that the leaf $b$ is $0$, taking into
account outside knowledge of the value of $g$. Let $q_\tau$ be the probability
that $b=0$ considering only leaves descending from $g$. Then with probability
$p_\tau$, $b$ will be measured to be $0$, in which case $S_{\tau+1}=q_\tau S\tau$.
With probability $1-p_\tau$ $b$ will be measured to be $1$, in which case
$S_{\tau+1}=(1-q_\tau) S\tau$. Notice that $S_\tau$ behaves similarly to
$A_\tau$, with $q_\tau$ in the recursive relation where $p_\tau$ is expected.
To obtain a lower bound on $S$, we will show that at each step
$\tau$, $S_\tau$ cannot decrease much more than $A_\tau$. In particular
we will show $p_\tau/q_\tau\leq 2$ for all $\tau$, which means at each step
$S_\tau$ decreases by at most twice as much as $A_\tau$. So after $T$ steps,
if $A_T<F$ with some probability, then $S<2^{-T}F$ with the same probability.

At step $\tau$ we have some outside knowledge of $v(t,g)$. Let $p_g(r)$
be the probability that $g=r$ considering all available information (thus it
reflects the true probability of obtaining $g=r$). Note $p_g(r)$ can change at
each $\tau$ as the outside information changes.) Let $p_{b|r}$ be the conditional
probability that a tree not excluded from $\mathcal{T}_{1,r}$ has $v(b)=0$.
$p_{b|r}$ is calculated based only on leaves evaluated in $t$ (i.e. using no outside
information). Our confidence level is at most $c_{1,1}$ (or otherwise there
is no need to evaluate further), so $|p_{root}(g,r)-1/2|\leq c_{1,1}/2$, which
is asymptotically small. Up to this small term, we can replace $p_{root}(g,r)$
with $1/2$. Then
\begin{align} \label{manyprobs}
\frac{p_\tau}{q_\tau}&=\frac{p_{g}(0)p_{b|0}+p_{g}(1)p_{b|1}}{p_{root}(g,0)p_{b|0}+p_{root}(g,1)p_{b|1}}\notag \\
&=\frac{p_{g}(0)p_{b|0}+p_{g}(1)p_{b|1}}{(p_{b|0}+p_{b|1})/2}\notag \\
&\leq 2.
\end{align}
So at each step $\tau$, we divide $S$ by at most an extra factor of 2 in relation
to $A_\tau$. We emphasize that here we assumed nothing about the value of $p_g(0)$.
Indeed it will become crucial later that for the purpose of bounding the decrease in
$S$, the values $p_g(0)$ can be anything.

Now we apply Lemma \ref{intlem}: at the beginning of the computation, we have
$p_{cat}(i,0)=p_{cat}(i,1)=1$ for all $1\leq i\leq\tilde{n}$, and so $S=2\tilde{n}$.
Thus we consider $A_0=2\tilde{n}$, $m=\beta=\lfloor (\log\tilde{n})/10\rfloor$ and
$F=\tilde{n}^{-3/10}$. Then $A_m<2\tilde{n}^{7/10}$ with probability less
than $2^mF=\tilde{n}^{-1/5}$. From our extra factor of $2$ above, over
the $\beta$ steps, we see that $S$ is smaller than $A_m$ by at most a factor
of $\tilde{n}^{1/10}$, so $S\lesssim 2\tilde{n}^{3/5}$ with probability less
than $\tilde{n}^{-1/5}$. In other words, $S\gtrsim2\tilde{n}^{3/5}$ with
probability greater than $1-\tilde{n}^{-1/5}$. We call this procedure the
{\it{division process}}. If $A_m\lesssim 2\tilde{n}^{7/10}$ (which will happen
with small probability), then we say the division process was {\it{lucky}}.

For the bound on $D$, we will look at how the $p_{cat}$ are updated after
each evaluation of a leaf. Immediately after the first leaf is evaluated,
all $p_{cat}$ become $1/2$ because for each category $i$ and root value $r$,
any leaf initially has an equal probability
of being $0$ or $1$. Now suppose a leaf $b$ is evaluated at some later point and
 $E$ is the set of leaves already evaluated.
We have $c^n$ leaves, and we can index each leaf using the labels $\{1,\dots,c\}^n$,
so for example, the labels of
two leaves sharing the same parent node differ only in the last coordinate.
Let $h$ be the maximum number of initial coordinates that agree
betwee $b$ and any element of $E$. Thus $h$ is the depth of the last node where the
path from the root to $b$ breaks off from the paths from the root to already evaluated nodes.
Let $\{e_j\}\in E$ be the leaves in $E$ that share $h$ initial coordinates with $b$. We
will now update each $p_{cat}(i,r)$, depending on the outcome of the
evaluation of $b$, and the position of $h$ relative to $i$.

There are three cases:
\begin{itemize}
\item $\bf{[i+n_0\leq h]}$ Here the faults occur at a greater height than the split at $h$.
Then $b$ and $\{e_j\}$ are both contained in a subtree of the form $t_{r,j}$,
so in such a tree, knowledge of a single node determines all other nodes.
Thus $b$ can be predicted from the values of $\{e_j\}$. If $b$ agrees with
the predicted values, then $p_{cat}(i,r)$ remains the same. Otherwise
$p_{cat}(i,r)=0$, since we learn that the faults must in fact come after
the split at $h$.
\item $\bf{[i>h]}$ Now $b$ is the first leaf we evaluate that descends from
some node $d$ at level $i$. $b$ has equal probability of being $0$ or $1$
(even if we know the value of $d$), so $p_{cat}(i,r)$ is reduced by half.
\item $\bf{[i\leq h<i+n_0]}$ Without knowing the details of $f_D$, we cannot
know what this part of the tree looks like. However, the increase in
$|p_{cat}(i,0)-p_{cat}(i,1)|$ is at most $1$ for each $i$. Since there are
$n_0$ $i$'s that satisfy the condition $i\leq h<i+n_0$, the total increase
in $D$ is at most $n_0$.
\end{itemize}
We see the only case that involves an increase in $D$ is the third,
with an increase of at most $n_0$ at
each evaluation. Since $D$ is $0$ at the beginning, and we have at most
$\beta$ evaluations, $D\leq\beta n_0$. Combining this with our bound on $S$,
we have
\begin{align}
c_{1,1}&\leq \frac{D}{S} \lesssim \frac{n_0\beta}{2\tilde{n}^{3/5}} \notag\\
p_{1,1}&\lesssim \tilde{n}^{-1/5}.
\end{align}
This is more stringent than what is required for Lemma \ref{mainlem}, so we have
completed the basis case.

\subsection{Induction Step for Proof of Lemma \ref{mainlem}}

\subsubsection{Exceptions} In the $k=1$ case, the bound on $p_k$ comes from
the probability of getting lucky during the division process. In the $k>1$ case,
there are several different ways to get ``lucky." We will call such occurrences
{\it{exceptions.}} We will first show the the probability
of getting an exception is low.
Then, in the final two sections, we will show that
as long as an exception does not occur, the confidence level of the classical algorithm
will be bounded according to Lemma \ref{mainlem}.

Consider a tree in $\mathcal{T}_k$, with root $g$ and nodes $a\in L$ at depth
$n$. Then the nodes in $L$ are the leaves of a tree in $\mathcal{T}_1$, and the
roots of trees in $\mathcal{T}_{k-1}$ by our construction of $\mathcal{T}_k$. We
call the tree in $\mathcal{T}_k$ the main tree and the trees rooted in $L$ subtrees.
We put no restrictions on the order that the classical algorithm evaluates leaves,
but to prove the lower bound, we will think of the algorithm as trying to find
the values of nodes in $L$, and from there, performing the division process on
the $\mathcal{T}_1$ tree as described in the $k=1$ case.

With this set up, we can describe the three types of exceptions.
Type I occurs when exceptions occur at two or more subtrees.
Type II occurs when the division process in the main tree is lucky, i.e. when
$A_\tau\ge 2\tilde{n}^{7/10}$ for some $\tau\le m$. This can only happen after
the division process has started (that is, after obtaining the second sure value
of a node at depth $n$. Assuming there are no type I exceptions, this must take
at least $\beta^{k-1}$ evaluations). After the start of the division process,
this occurs with probability at most $\tilde{n}^{-1/5}$.
Type III occurs when an exception occurs in one subtree, but with special
conditions, which will be explained below. It will be shown in Section
\ref{Sbound} that $S$ cannot be too small if none of the above has happened.

\subsubsection{Bound on $p_{k}$}

In this section we will bound the probability of learning the value at the root
of the tree when an exception occurs.

First we will define Type III exceptions.

\begin{definition} \label{IIIdef}
Let $b$ be a leaf descending from $a\in L$, with $g$ the root of the tree.
If $v(b)={r'}$ we know we will have an exception at the subtree rooted at $a$,
while if $v(b)=\bar{r}'$ then there is no exception. Let $p_{partial(a)}$ be
the probability that $v(b)=r'$ considering only information from other evaluated
leaves descending from $a$. Let $p_{partial(g)}$ be the probability that
$v(b)=r'$ considering only information from other evaluated leaves descending
from $g$ (i.e. no outside knowledge). Then a Type III exception occurs if $b$
is measured to be $r'$, and if $p_{partial(g)}/p_{partial(a)}<2\beta^{-3}.$
We call $p_{partial(g)}/p_{partial(a)}$ the bias ratio.
\end{definition}
\noindent We will see in Section \ref{Sbound} why Type III exceptions harm the confidence level.

Again, the three types of exceptions are:
\begin{itemize}
\item{Type I} are when exceptions occur at two subtrees rooted in $L$.
\item{Type II} are when the division process is lucky.
\item{Type III} is described in Def. \ref{IIIdef}, and is when an exception
occurs with a bias ratio
less than $2\beta^{-3}.$
\end{itemize}


We will only consider the case $l=k$, since if $l<k$ nodes are evaluated,
the probability of having an exception cannot be more than if $l=k$.
In proving this, we will use a generalization of $T_k$, that is more flexible
to use in an induction.

\begin{definition} \label{tildetree}
An element of the set $\tilde{\mathcal{T}}_k$ is a collection of $c^n$ $c$-ary
trees from $\mathcal{T}_{k-1}$. Each tree in the collection is called a subtree.
(We call this a subtree because it corresponds roughly to a subtree in $\mathcal{T}_k$,
but notice that subtrees in $\tilde{\mathcal{T}}_k$ are not connected to each other
as a part of a larger tree structure as in $\mathcal{T}_k$. In both $\mathcal{T}_k$
and $\tilde{\mathcal{T}}_k$, a subsubtree is a tree rooted at a depth $n$ below
the root of a subtree.) All roots of subtrees in $\tilde{\mathcal{T}}_k$ have
{\it{a priori}} probabilities that are arbitrary, time dependent (the priors can
change after each evaluation of a leaf on any subtree), and unrelated to each other.
For convenience, we call the whole collection of subtrees in $\tilde{\mathcal{T}}_k$
a tree. Then a Type I exception in such a tree is defined to be two exceptions at two
different subtrees, where exceptions at subtrees are normal exceptions in
$\mathcal{T}_{k-1}$. Let $\tilde{p}_{partial(a)}(r)$ be the probability that a leaf
$b$ descending from a root $a$ of a subtree in $\tilde{\mathcal{T}}_k$ has value $r$ based only on other evaluated
leaves descending from $a$. Let $\tilde{p}_{partial(g)}(r)$ be the probability
that a leaf $b$ descending from a root node $a$ has value $r$ including the
{\it{a priori}} probabilities associated with $a$. Then a Type III exception is
defined to be an exception at a subtree with
$\tilde{p}_{partial(g)}(r)/\tilde{p}_{partial(a)}(r)<2\beta^{-3}$.
We call $\tilde{p}_{partial(g)}(r)/\tilde{p}_{partial(a)}(r)$ the bias ratio for
trees in $\tilde{\mathcal{T}}_k$. Exceptions in subtrees are exceptions for
$\mathcal{T}_{k-1}$ as defined above. As in a normal $\mathcal{T}_{k-1}$ tree,
if an algorithm makes $\beta^{k-1}$ or more queries at a subtree of $\tilde{\mathcal{T}}_k$, the value
of the root of the subtree is given, and exceptions can no longer occur at that subtree.
\end{definition}

Note that in the structure $\tilde{\mathcal{T}}_k$, binary values are associated to each node of the subtrees $\mathcal{T}_{k-1}$,
but nowhere else. In particular there is no such thing as the value of the root of
$\tilde{\mathcal{T}}_k$. Note also that there are no type II exceptions for
$\tilde{\mathcal{T}}_k$. The numberings for exceptions are used in analogy with exceptions for
$\mathcal{T}_k$.

We may consider $\tilde{\mathcal{T}}_k$ as a black box that accepts as input the
location of the next leaf to query, and outputs the value of that leaf according
to the priors, as well as a new set of priors for all subtrees.
At some time, suppose that the priors that a certain subtree's root
has value 0 is $p$ and has value 1 is $1-p$. Further, suppose that based on previous
outputs of the black box, the distribution at that subtree of 0-valued subtrees is $\tilde{T}'_{k-1,0}$
and the distribtion of 1-valued subtrees is $\tilde{T}'_{k-1,1}$.
Then the value returned by the black
box must be chosen according to the distribution $p\tilde{T}'_{k-1,0}+(1-p)\tilde{T}'_{k-1,1}$.
 To prove that an algorithm
on $\mathcal{T}_k$ has a low probability of exception, we will need to prove upper bounds on the
probability of exceptions for an algorithm running on $\tilde{\mathcal{T}}_k$.
For this purpose, we can assume that the algorithm on $\tilde{\mathcal{T}}_k$ is specifically trying to get
such an exception, and that the priors can change in a way that is beneficial
to the algorithm.



Our strategy in proving the bound on the probability of exception in a tree
from $\mathcal{T}_k$ is as follows: we will first show that we can map an
algorithm $\mathcal{A}$ acting on a tree from $\mathcal{T}_k$ to an algorithm
$\tilde{\mathcal{A}}$ acting on a tree from $\tilde{\mathcal{T}}_k$. The
mapping is such that the probability of getting a Type I or Type III exception
is the same in both algorithms. Next we will map $\tilde{\mathcal{A}}$ to an
algorithm $\tilde{\mathcal{B}}$, which acts on a smaller set of subtrees than
$\tilde{\mathcal{A}}$. We will show that the probability that $\tilde{\mathcal{B}}$
has exceptions is similar to the probability that $\tilde{\mathcal{A}}$ has
exceptions. Using an inductive proof, we will then bound the probability of
$\tilde{\mathcal{B}}$ having an exception, which in turn bounds the probability
of $\tilde{\mathcal{A}}$ having an exception, which then bounds the probability
of $\mathcal{A}$ getting a Type I or III exception. Type II exceptions are
treated slightly differently throughout the analysis, using the fact that
the probability of a Type II exception in a tree from $\mathcal{T}_k$ is
always bounded by $\tilde{n}^{-1/5}$ once the division process has started
in that tree.

Consider a mapping from a tree in $\mathcal{T}_k$ to a tree in
$\tilde{\mathcal{T}}_k$ where any time an algorithm $\mathcal{A}$ evaluates
a leaf on a subtree from a distribution $\mathcal{T}_k$, we copy the value of
the leaf to the corresponding leaf on the corresponding subtree from
$\tilde{\mathcal{T}}_k$. (Note $\tilde{\mathcal{T}}_k$ was created so that
a tree in $\mathcal{T}_k$ has the same number of subtrees as a tree in
$\tilde{\mathcal{T}}_k$.) Set the {\it{a priori}} probabilities at a root
$\tilde{a}$ of a subtree in $\tilde{\mathcal{T}}_k$ to equal that
of the corresponding node $a$ at depth $n$ in $\mathcal{A}$ (this guarantees
that the values returned by $\tilde{\mathcal{T}}_k$ are in the right
distribution). We call $\tilde{\mathcal{A}}$ the algorithm acting on
$\tilde{\mathcal{T}}_k$. Notice that anytime a Type I exception occurs in $\mathcal{A}$, it also
occurs in $\tilde{\mathcal{A}}$ since if there is an exception at a subtree
in one, there is an identical subtree on the other that also has an exception.
Also, anytime $\mathcal{A}$ has a Type III exception, algorithm $\tilde{\mathcal{A}}$
will also have a Type III exception, since the priors are the same. So we can
conclude that the probability that $\mathcal{A}$ has a Type I or Type III exception
is identical to the probability that $\tilde{\mathcal{A}}$ has Type I
or Type III exceptions.

Our strategy will be to bound the probability that a tree from $\tilde{\mathcal{T}}_k$
can have an exception of Type I or III, and then we will use the bound that
a Type II exception happens with probability $\tilde{n}^{-1/5}$ once the division
process starts (which it will for $\beta^k$ evaluations), for the same reason
we get lucky with probability $\tilde{n}^{-1/5}$ in the $k=1$ case.

Let $\tilde{\mathcal{A}}$ be as defined above. To get the bounds on the probability
of exception, we will simulate $\tilde{\mathcal{A}}$ using an algorithm $\tilde{\mathcal{B}}$
which acts on $\beta$ trees of type $\tilde{\mathcal{T}}_{k-1}$. Suppose
$\tilde{\mathcal{A}}$ evaluates leaves descending from a subtree rooted at the node
$a_i$, and furthermore, evaluates leaves descending from subsubtrees at depth $n$
below $a_i$, rooted at the nodes $(a_{i1},a_{i2},\dots)$. Then
$\tilde{\mathcal{B}}$ randomly picks one of its $\beta$ trees on which to
simulate $a_i$, and each time $\tilde{\mathcal{A}}$ evaluates on a new subsubtree
of $a_i$ rooted at $a_{ij}$, $\tilde{\mathcal{B}}$ randomly chooses a new subtree
in the tree it has chosen to simulate $a_i$, and one that isn't yet being used
for simulation. Then all leaves evaluated on a subsubtree by $\tilde{\mathcal{A}}$
are mapped onto the leaves of the corresponding subtree in $\tilde{\mathcal{B}}$.
This setup will always be possible because the number of evaluated leaves is
small compared to the number of nodes at depth $n$. The result is that each subtree
in $\tilde{\mathcal{A}}$ is mapped (not necessarily injectively) to one of the $\beta$
$\tilde{\mathcal{T}}_{k-1}$ trees, while the sub-subtrees in $\tilde{\mathcal{A}}$ are mapped
injectively to subtrees of the appropriate $\tilde{\mathcal{T}}_{k-1}$.

Our strategy will be to show that $\tilde{\mathcal{B}}$ has Type I and Type III exceptions
at least as often as $\tilde{\mathcal{A}}$, and then to bound the probability that
$\tilde{\mathcal{B}}$ has an exception. Notice that each of the $\beta$ trees in $\tilde{\mathcal{B}}$ can have
at most $\beta^k$ evaluations on it. So we will prove the following lemma:

\begin{lemma} \label{kplus1}
Suppose an algorithm makes less than $\beta^{k+1}$ evaluations on a tree from
$\tilde{\mathcal{T}}_{k}$. Then the probability of getting
one exception on a subtree is $\lesssim (\beta^{3}+3\beta^2)\tilde{n}^{-1/5}$.
The probability of getting a Type I exception is
$\lesssim(\beta^{3}+3\beta^2)^2\tilde{n}^{-2/5}$. The probability of getting a
Type III exception is
$\lesssim (1+3\beta^{-1})2\tilde{n}^{-1/5}$.

\end{lemma}

Note that we are given $\beta$ times more total evaluations than in Theorem \ref{mainthm}. However, we will 
give the algorithm the added advantage that once $\beta^{k-1}$ queries
are made on a tree from $\mathcal{T}_{k-1}$, the algorithm will just return the
value of the root of the subtree. This added information will only be an advantage. Thus,
the number
of evaluations on each subtree is bounded by $\beta^{k-1}$. Essentially we are trying to show that the
probability of exception also increases by at most a factor of $\beta$.
To prove Lemma \ref{kplus1}, we require one further lemma:

\begin{lemma} \label{dividethepot}
Suppose we have a finite sequence of real numbers $\{b_i\}$, such that each
$b_i\leq \beta^{k-1}$ and $\sum_i b_i\leq\beta^{k+1}$. Construct a subset $W$
of $\{b_i\}$ by choosing to place each $b_i$ into $W$ with probability
$\beta^{-2}$. Then the probability that $\sum_{b_i\in W}b_i>\beta^k$ is of order
$1/\beta!$.
\end{lemma}
\begin{proof}
First normalize by letting $b_i'=b_i/\beta^{k-1}$. Then we can
describe the sum of elements in $W$ as a random variable $X$, which is itself the
sum of random variables $X_i$. The value of $X_i$ is chosen to be $0$ if
$b_i \not\in W$ and $b_i'$ if $b_i\in W$. Then the mean $\bar{X_i}$ is less than
$\beta^{-2}$ and $\sum\bar{X}_i\leq1$. Since on average, at most one of the
$X_i$'s will be $1$, $W$ is very close to a Poisson distribution with parameter
at most 1. Then the probability that $X>\beta$ is at most of order $1/\beta!$.
\end{proof}

Now we will prove Lemma \ref{kplus1}:
\begin{proof}
The proof will follow from induction on $k$. For $k=1$, there are no subtrees,
so the probability of exceptions on subtrees is 0.

For the inductive case, we assume we have an algorithm $\tilde{\mathcal{D}}$ acting on
a tree from $\tilde{\mathcal{T}}_{k}$, such that $\tilde{\mathcal{D}}$ can evaluate
leaves descending from up to $\beta^{k+1}$ subtrees. Notice that once
$\tilde{\mathcal{D}}$ evaluates more than $\beta^{k-1}$ leaves of any subtree, an exception can
no longer occur on that subtree and in fact the algorithm is given the value at
the root of the subtree, so we can stop simulating that subtree.
It takes $\beta^{k-2}$ evaluations on any subtree to begin the division process
in that subtree, so $\tilde{\mathcal{D}}$ with its $\beta^{k+1}$ total evaluations has
$\beta^3$ opportunities to get an exception of Type II in a subtree.
Thus the probability of getting an exception of Type II in one subtree is at most
$\beta^3\tilde{n}^{-1/5}$.

To determine the probability that $\tilde{\mathcal{D}}$ triggers Type I and
III exceptions in subtrees, we will simulate using a new algorithm $\tilde{\mathcal{F}}$.
$\tilde{\mathcal{F}}$ evaluates on $\beta^2$ trees of type $\tilde{\mathcal{T}}_{k-1}$. The simulation
will work similarly as above, except $\tilde{\mathcal{F}}$ can evaluate on $\beta^2$ trees
instead of $\beta$ trees.

Now suppose $\tilde{\mathcal{D}}$ produces an exception of Type I or Type III on a subtree
rooted at $a$. We show that this implies an exception in the corresponding tree in
$\tilde{\mathcal{F}}$. A type I exception on $a$ means there are two exceptions on subtrees
of $a$ (of type $\mathcal{T}_{k-2})$. Since there is an injective correspondence on subtrees
of type $\mathcal{T}_{k-2}$, there are exceptions on two subtrees of the $\tilde{\mathcal{T}}_{k-1}$
that $a$ is assigned to. The argument for type III exceptions is similar. If less than
$\beta^{k}$ evaluations are made on the corresponding tree in $\tilde{\mathcal{F}}$,
then we can use the probabilities from the inductive assumption to bound the likelihood
of exceptions occurring. Furthermore, using Lemma \ref{dividethepot} we will show
that the probability that more than $\beta^k$ evaluations are made on any one
of the $\beta^2$ trees used in $\tilde{\mathcal{F}}$ is at most of order $1/\beta!$.

Let $b_i$ in Lemma \ref{dividethepot} be the number of evaluations $\tilde{\mathcal{D}}$ has made
on a subtree rooted at $a_i$. Notice that all $b_i$ are less than
$\beta^{k-1}$, as required since $\tilde{\mathcal{D}}$ makes no more than $\beta^{k-1}$
evaluations on any one node in $L$. $W$ then corresponds to the amount of
evaluations on one of the $\beta^2$ trees that $\tilde{\mathcal{F}}$ evaluates. By
Lemma \ref{dividethepot} the probability that the number of evaluations on one
of $\tilde{\mathcal{F}}$'s trees is larger than $\beta^k$ is at most $1/\beta!$.
Since $\beta!$ is superpolynomial in $\tilde{n}$, we can ignore the probability
of this occurring, even after multiplying by $\beta^2$ for each of the possible
trees used by $\tilde{\mathcal{F}}$.

A Type I or Type III exception at a subtree $\tilde{\mathcal{D}}$  creates a
corresponding exception at a root of one of $\tilde{\mathcal{F}}$'s trees.
Thanks to Lemma \ref{dividethepot}, we can use the inductive
assumptions, so the probability of one exception of Type I or III occurring on
a tree in $\tilde{\mathcal{F}}$ is $\lesssim (1+3\beta^{-1})2\tilde{n}^{-1/5}+
(\beta^{3}+2\beta^2)^2\tilde{n}^{-2/5}\lesssim3\tilde{n}^{-1/5}$. Since there are
$\beta^2$ trees used by $\tilde{\mathcal{F}}$, the total probability of $\tilde{\mathcal{F}}$
having 1 exception of Type I or III is at most $3\tilde{n}^{-1/5}\beta^2$,
and therefore this is also a bound on the probability of $\tilde{\mathcal{D}}$
having an exception. Combining
this with the probability of $\beta^3\tilde{n}^{-1/5}$ of getting a Type II exception
in a subtree in $\tilde{\mathcal{D}}$, we have the first bound in Lemma \ref{kplus1}.

If $P$ is the the probability
that $\tilde{\mathcal{D}}$ has Type I or III exceptions on two of its subtrees (leading to a Type I exception),
then $\tilde{\mathcal{F}}$ will also have exceptions on two of its trees with probability at least $P(1-\beta^{-2})$
since it will have two exceptions unless it happens to simulate both of the subtrees with exceptions
on the same tree, which occurs with probability $\beta^{-2}$. Using the inductive assumption,
the probability of getting two Type I or III exceptions is the square of the probability of getting
a single exception, times the number of ways that two trees can be chosen
among $\beta^2$ trees:
\begin{equation}
P(1-\beta^{-2})\lesssim4\tilde{n}^{-2/5}\frac{\beta^2(\beta^2-1)}{2}\rightarrow P\lesssim4\tilde{n}^{-2/5}\beta^4.
\end{equation}
Combining this probability with the probability of a Type II error, we obtain the
 second inequality in Lemma \ref{kplus1}.

For the final inequality in Lemma \ref{kplus1}, regarding exceptions of Type III, we
will simulate the evaluations that $\tilde{\mathcal{D}}$ makes
using another algorithm $\tilde{\mathcal{C}}$ that is also evaluating
on a tree from $\tilde{\mathcal{T}}_k$. The value of any leaf
that $\tilde{\mathcal{D}}$ evaluates is mapped to the corresponding leaf in
$\tilde{\mathcal{C}}$. However, the {\it{a priori}} probabilities of root of subtrees
in $\tilde{\mathcal{C}}$ are adjusted so that any time
$\tilde{\mathcal{D}}$ is about to get a Type III exception, the bias ratio in
$\tilde{\mathcal{C}}$ goes to 1.  To get a bias ratio
of one, simply make the {\it{a priori}} probabilities equal.
We are free to change the {\it{a prior}}
probabilities in this way and still have the tree be in $\tilde{\mathcal{T}}_k$,
 and therefore, the inductive assumption, as well as the first two
 bounds in
 Lemma \ref{kplus1} must hold for $\tilde{\mathcal{C}}$ since it holds for
 all trees in $\tilde{\mathcal{T}}_k$.

Note that if an exception at a subtree rooted at $a$ is possible in $\tilde{\mathcal{D}}$, then it will also be possible in
$\tilde{\mathcal{C}}$. If an exception can happen in $\tilde{\mathcal{D}}$, then conditioning on $v(a)=0$
or $v(a)=1$, there is a non-zero probability that $b$, a leaf descending from $a$, will trigger an exception at $a$. So when we look
at an even likelihood of $v(a)=0$ and $v(a)=1$, there will always be some probability that $b$
will trigger an exception at $a$ in $\tilde{\mathcal{C}}$.


By Lemma \ref{kplus1},
$\tilde{\mathcal{C}}$'s probability
of getting one exception at $a$ is just $(\beta^3+3\beta^2)\tilde{n}^{-1/5}$.
Changing the
bias ratio changes the probability that there is an exception at $b$. Let $p_{partial(a)}$
 be the probability that there is an exception at $a$ due to $b$, only considering information
 from leaves descending from $a$. Then using Eq. \ref{manyprobs} and the definition of the bias
 ratio, in $\tilde{\mathcal{C}}$ the probability of getting
an exception at $a$ is less than $2p_{partial(a)}$.
So $2p_{partial(a)}\lesssim (\beta^3+3\beta^2)\tilde{n}^{-1/5}$.
Now in $\tilde{\mathcal{D}}$, by the same reasoning,
the probability of getting an exception at $a$ is less than $4p_{partial(a)}/\beta^3$.
But we can plug in for $2p_{partial(a)}$ to get the probability of exception is $\lesssim
(1+3\beta^{-1})2\tilde{n}^{-1/5}$.
\end{proof}

Now we can put this all together to get the bound for on $p_k$. As described in the pargraphs following
Def. \ref{tildetree},
we simulate
the algorithm $\tilde{\mathcal{A}}$ using an algorithm $\tilde{\mathcal{B}}$, where $\tilde{\mathcal{B}}$
makes evaluations on only $\beta$ trees from $\tilde{\mathcal{T}}_{k-1}$. First we will determine
the probability of getting each type of exception on one subtree in $\tilde{\mathcal{A}}$.
Since it takes $\beta^{k-2}$ evaluations to begin the division process in a
subtree, and we have $\beta^{k}$ total evaluations, $\tilde{\mathcal{A}}$ has at most
$\beta^2$ attempts to get one exception of Type II in a subtree, for a total
probability of $\beta^2\tilde{n}^{-1/5}$.

From Lemma \ref{kplus1}, the $\beta$ trees evaluated by $\tilde{\mathcal{B}}$ each
have a probability of $\lesssim(\beta^3+2\beta^2)^2\tilde{n}^{-2/5}$ of having
an exception of Type I, and a probability of $\lesssim (1+3\beta^{-1})2\tilde{n}^{-1/5}
$ for Type III, which we know is at least as large as the probability of $\tilde{\mathcal{A}}$
itself having a corresponding exception on a subtree.

Combining all possibilities, we get that the probability of an exception of any
type on one subtree in $\tilde{\mathcal{A}}$ is asymptotically smaller than
\begin{align}
\beta^2\tilde{n}^{-1/5}+ \beta((1+3\beta^{-1})2\tilde{n}^{-1/5}+
(\beta^3+2\beta^2)^2\tilde{n}^{-2/5}
)\notag \\
\lesssim\beta^3\tilde{n}^{-1/5}.&
\end{align}

Finally, we can determine the asymptotic probability of getting Type I and III exceptions
each  in $\tilde{\mathcal{A}}$:
\begin{itemize}
\item Type I: Using the simulation, this is bounded by the square of the probability of getting a single
exception, times the number of ways two subtrees can be chosen from
$\tilde{\mathcal{B}}'s$ $\beta$ available trees, giving $\beta^6\tilde{n}^{-2/5}\beta(\beta-1)/2$.
\item Type III: By the same reasoning as at the end of Lemma \ref{kplus1},
the probability of getting an exception with bias ratio $2\beta^{-3}$ is at
most $2\beta^{-3}\beta^3\tilde{n}^{-1/5}=2\tilde{n}^{-1/5}$.
\end{itemize}

Now our original algorithm $\mathcal{A}$ acting on a tree in $\mathcal{T}_k$ has
a probability of Type I and III exceptions that is
at most that of $\tilde{\mathcal{A}}$, but it additionally has
a $\tilde{n}^{-1/5}$ probability of a Type II exception.
Combining these three probabilities, we have
\begin{equation} \label{pkk}
p_{k}\lesssim\tilde{n}^{-1/5}+2\tilde{n}^{-1/5}+\beta^6\tilde{n}^{-2/5}\beta(\beta-1)/2
\lesssim3\tilde{n}^{-1/5}.
\end{equation}

\subsubsection{Bounding S} \label{Sbound}

Now that we've seen that the probability of an exception is low, we will bound the confidence
level assuming no exceptions occur. As in the $k=1$ case, we will put bounds on $S$ and
$D$ rather than on the confidence level directly.
We use an iterative proof, so we assume that Lemma \ref{mainlem} holds for
$\mathcal{T}_{k-1}$. In fact, to simplify the argument, if $\beta^{k-1}$
leaves descending from a node $a\in L$ have been evaluated, or if an exception
occured at $a$, we will give the
algorithm the correct value of $a$. The algorithm can do no worse with this information.

By induction, we can apply the confidence bound $c_{k-1,l}$ to nodes in $L$
since they are the roots of trees from $\mathcal{T}_{k-1}$. While this gives
us information about $p_{root}(a,r),$ what we really care about is
$p_{root}(g,r)$. To calculate the latter, we define $p_{tree}(t)$ for $t\in\mathcal{T}_1$ to be
the probability that a randomly chosen tree from $\mathcal{T}_k$ with its first
$n$ levels equal to $t$ is not excluded. Then
\begin{equation} \label{eqtree}
p_{tree}(t) =\prod_{a\in L}p_{root}(a,v(t,a)),
\end{equation}
and we let $p_{cat}(i,r)$ be the weighted average of $p_{tree}(t)$ for all
$t\in\mathcal{T}_{1,r,i}$. This weighting is due to uneven weightings of trees
in $\mathcal{G}_r$. Trees in $\mathcal{T}_1$ containing less likely trees from
$\mathcal{G}_r$ should have less weight than those containing more likely ones.

From $p_{tree}$, we can obtain a bound on the change in $p_{cat}(g,r)$ due to
a change in a $p_{root}(a,r)$. Suppose due to an evaluation on a leaf, $p_{root}(a,0)$
changes to $p_{root}(a,0)'$. Then every $p_{tree}(t)$ where $v(t,a)=0$ will change
by a factor of $p_{root}(a,0)'/p_{root}(a,0)$. Now each $p_{cat}$ is a weighted average of
some trees with $0$ at $a$, and some with $1$ at $a$. So at most, the multiplicative change
in each $p_{cat}$ due to $p_{root}(a,0)$ is equal to
$p_{root}(a,0)'/p_{root}(a,0)$. Then the change in $\log S$ is
at most $\log(p_{root}(a,0)'/p_{root}(a,0))$, which is just the change in $\log p_{root}(a,0)$.
We can go back and repeat the same logic with $a=1$, to get that the total change in $\log S$ is at most the
change in $\log p_{root}(a,0)$ plus the change in $\log p_{root}(a,1).$
This bound is only useful if $a$ is not learned with certainty.

Note that the confidence level does not change if we renormalize
our probabilities because the confidence level is a ratio of probabilities.
Since $p_{tree}(t)$ and $p_{root}(a,r)$ may be extremely small, we will
renormalize $p_{root}(a,r)$ to make the probabilities easier to work with:
we renormalize so that $p_{root}(a,0)+p_{root}(a,1)=2$.
We will just need to be careful when we apply the results of the $k=1$ case,
that we take the new normalization into account.

We will now prove several lemmas we will use in bounding $S$.
\begin{lemma} \label{Alemm}
Suppose no exception has occurred at $a\in L$, and less than $\beta^l$ leaves descending from
$a$ have been evaluated, and $p_{root}(a,r)_i$ is the value of
 $p_{root}(a,r)$. At some later time, such that still,
 no exception has occurred at $a\in L$, and less than $\beta^l$ leaves descending from
$a$ have been evaluated,
 $p_{root}(a,r)_f$ is the value of  $p_{root}(a,r)$. Then the change in
$\log p_{root}(a,r)=\log p_{root}(a,r)_f-\log p_{root}(a,r)_i$
obeys $\Delta\log p_{root}(a,r)\lesssim 2c_{k-1,l}$.
\end{lemma}

\begin{proof}
By induction, if no exception occurs, and if less than $\beta^l$ leaves
descending from $a\in L$ have been evaluated, then using our normalization
condition we have
\begin{align} \label{newconf}
\frac{|p_{root}(a,0)-p_{root}(a,1)|}{p_{root}(a,0)+p_{root}(a,1)}
=|p_{root}(a,r)-1|
\lesssim c_{k-1,l}
\end{align}
for r $\in\{0,\}1$. Then with $p_{root}(a,r)_f$ and $p_{root}(a,r)_i$ as defined above,
 $p_{root}(a,r)_f=A\cdot p_{root}(a,r)_i$, where in order to
satisfy Eq. \ref{newconf},
\begin{align}
A\in\left[\frac{1-c_{k-1,l}}{1+c_{k-1,l}}, \frac{1+c_{k-1,l}}{1-c_{k-1,l}} \right]
&\Rightarrow|\log A|\lesssim 2c_{k-1,l} \notag \\
&\Rightarrow \Delta\log p_{root}(a,r)\lesssim 2c_{k-1,l}.
\end{align}
\end{proof}
We will use Lemma \ref{Alemm} when we determine how $S$ changes when less than
$\beta^{k-1}$ leaves from a node have been evaluated.

Type III exceptions are bad because they cause $S$ to be divided by a much larger
factor than $A_\tau$ during the division process on the main tree, where $A_\tau$
keeps track of the true probabilities during the division process, as in the $k=1$
case.

\begin{lemma} \label{lemmaIII}
If there is an exception in a subtree, but it is not a Type III exception, the the
 probability that $S$ changes and the amount that $S$ changes
differ by at most a factor
of $4\beta^3$.
\end{lemma}

\begin{proof}

We will prove the contrapositive.
Without loss of generality, suppose $v(b)=0$ gives
an exception at $a$, and if $v(b)=0$, then we know $v(a)=0$ with high probability.
Let $q_r$ be the probability that $v(a)=r$ including all possible information.
Let $q_r'$ be the probability that $v(a)=r$ based only on evaluated leaves descending
from $g$, so no outside information.
Let the probability that $v(b)=0$ conditioned on $v(a)=r$ be $p_{b|r}$.

Conditioning on $v(b)=0$, the probability that $S$ changes
is equal to the probability that
$v(a)=0$ including all possible information, which by Bayes Rule is
$q_0p_{b|0}/(q_0p_{b|0}+q_1p_{b|1})$. Now the factor by which $S$ changes
is simply $q_0'$, since $S$ only depends on nodes descending
from $g$. To obtain a contradiction, we assume that these two terms differ by a lot. In particular, we assume

\begin{align} \label{newmanyprob1s}
q_0'\left(\frac{q_0p_{b|0}+q_1p_{b|1}}{q_0p_{b|0}}\right)<\frac{1}{4\beta^3}
\end{align}

We can rewrite this as
\begin{align} \label{newmanyprob2s}
q_0'\left(1+\frac{q_1p_{b|1}}{q_0p_{b|0}}\right)<\frac{1}{4\beta^3}
\end{align}

\noindent This tells us that (i) $q_0'<1/(4\beta^3)$ . Furthermore, using Eq.
\ref{manyprobs}, we know $q_0'\geq q_0/2$. Using this,
and the fact that $q_0+q_1=1$, we can write Eq. \ref{newmanyprob2s} as:

\begin{align} \label{newmanyprob3s}
\left(\frac{1+(q_1p_{b|1})/(q_0p_{b|0})}{1+\frac{q_1}{q_0}}\right)<\frac{1}{2\beta^3}
\end{align}

\noindent From this equation, we have the restriction $q_1/q_0>2\beta^3$ ($2\beta^3\gg1$ so
 we can ignore the $1$ in the denominator), and (ii) $p_{b|1}/p_{b|0}<1/(2\beta^3)$.

Now let's consider $p_{partial(g)}/p_{partial(a)}$:

\begin{align} \label{bratio}
\frac{p_{partial(g)}}{p_{partial(a)}}=
\frac{p_{b|0}q_0'+
p_{b|1}q_1'}{(p_{b|0}+p_{b|1})/2}
\end{align}

\noindent From (ii), let $p_{b|0}/p_{b|1}=2\beta^3M_1$, where $M_1>1$. From (i), let
 $q_0'=1/(4\beta^3M_2)$ where $M_2>1$, so then
$q_1'=1-1/(4\beta^3M_2)$. Then

\begin{align} \label{bratio2}
\frac{p_{partial(g)}}{p_{partial(a)}}=\frac{2\beta^3M_1/(4\beta^3M_2)+(1-1/(4\beta^3M_2))}
{\beta^3M_1+1/2}.
\end{align}

Since $\beta^3M_1\gg1/2$, we have

\begin{align} \label{bratio3}
\frac{p_{partial(g)}}{p_{partial(a)}}=\frac{1}{2\beta^3M_2}+\frac{1}{\beta^3M_1}-\frac{1}{4\beta^3M_1}
<\frac{2}{\beta^3}.
\end{align}
\end{proof}

We can now prove the following lemma:
\begin{lemma} \label{lemmaS} For the case (k,l) in Lemma \ref{mainlem}, if an
exception on the main tree does not occur, then if an algorithm evaluates less
than $\beta^l$ leaves, $S\gtrsim \tilde{n}^{3/5}\beta^{-3}/2$ if $k=l$ and
$S\gtrsim2\tilde{n}$ otherwise.
\end{lemma}
\begin{proof}
For the case $l=k$, we know with high probability at most $\beta$ nodes in $L$
from directly evaluating descendents of those nodes.
(We learn ($\beta-1$) nodes from evaluating $\beta^{k-1}$ leaves on each of
$\beta-1$ subtrees, and 1 from one possible exception on a subtree). Thus at
most $\beta$ times during the algorithm, a $p_{root}(a,r)$ is set to 0, and its
pair is set to 2. When the $\beta^{(k-1)}$'th leaf descending from a node $a$
is evaluated, this can be thought of as a query for $a$. So the division process
proceeds identically as in the case $k=1$, except for two differences.
First, when the $j^{th}$ node in $L$
is learned with certainty, we must multiply $S$ by a normalizing factor $F_j$. This
is due to renormalizing $p_{root}(a,0)+p_{root}(a,1)=2$ in the $k>1$ case,
whereas in the $k=1$ case, when a node $a$ at depth $n$ is learned, $p_{root}(a,0)+p_{root}(a,1)=1$.
(Note $F_j>1$ for all $j$, and further, we don't actually
care what the $F_j$ are, since $D$ will be multiplied by the same $F_j's$.)
Second, by Lemma \ref{lemmaIII}, the one possible exception in a subtree causes $S$ to decrease
by at most an extra factor of $4\beta^3$ as compared to the
$k=1$ case. Combining these terms, we get that the
decrease in $\log S$, $\Delta \log S $, is at most $(2\log\tilde{n})/5+\log4\beta^3-\sum_j\log F_j$.

Next we need to estimate the decrease in $\log S$ due to nodes in $L$ where less
than $\beta^{k-1}$ of their leaves have been evaluated. There are $\beta$ times
in the algorithm when we learn the value of a node in $L$. There are then
$\beta+1$ intervals, between, before, and after those times. We will determine
the maximum amount that $\log S$ can change during a single interval, and
then multiply by $\beta+1$ to account for the change during all intervals. We will be considering
intervals because we will be using Lemma \ref{Alemm}, which only applies in periods
in between when the values of nodes in $L$ are learned.

For a given interval, let $M(a)$ be the number of leaves descending from $a\in L$
that have been evaluated. Then for $1\leq j\leq l$ let
$K_j=\{a\in L:\beta^{j-1}\leq M(a)<\beta^j\}$. Because less than $\beta^k$
leaves have been evaluated in total, we have the bound $|K_j|\leq\beta^{k-j+1}.$

We will bound the change in $S$ using the relation between $p_{root}$, $p_{tree}$,
and $p_{cat}$ described in the paragraph following Eq. \ref{eqtree} (as opposed
to the bound used in Lemma \ref{intlem}). This bound
applies regardless of the order of evaluations, so we are free to reorder
to group together evaluations
descending from the same node in $L$. This allows us to use Lemma \ref{Alemm}:
each of the nodes in $K_j$ can change $\log S$ by at most $4c_{k-1,j}$
 ($2c_{k-1,j}$ from $p_{root}(a,0)$ and $2c_{k-1,j}$ from $p_{root}(a,1)$). Since
there are at most $\beta^{l-j+1}$ nodes in $K_j$, the nodes in $K_j$ will change
$\log S$ by at most $4\beta^{l-j+1}c_{k-1,j}.$
Summing over the $K_j$ in all $(\beta+1)$ intervals, we have,
\begin{align}
\Delta\log S\lesssim &4(\beta+1)\notag \\
&\times(c_{k-1,k-1}\beta^2+c_{k-1,k-2}\beta^3+\dots+c_{k-1,1}\beta^{k+1})
\end{align}
Combining with the terms from the division process and subtree exception, we obtain
\begin{align}
\Delta\log S\lesssim&\frac{2\log\tilde{n}}{5}+\log(4\beta^3)-\sum_j\log F_j\notag \\
&+
4(\beta+1)(c_{k-1,k-1}\beta^2+\dots+c_{k-1,1}\beta^{k+1})
\end{align}
Since this is an induction step, we can plug in values for $c_{k-1,j}$ from Lemma
\ref{mainlem}, which, using the asymptotic definition, gives
\begin{equation}
\Delta\log S\lesssim \frac{2\log\tilde{n}}{5}+\log(4\beta^3)-\sum_j\log F_j
\end{equation}

Finally, using the fact that initially $S=2\tilde{n}$, we find
$S\gtrsim(\Pi_jF_j\tilde{n}^{3/5}\beta^{-3})$.

For the case $l<k$ by similar reasoning as above, we can reorder to put the
evaluations of leaves from the subtree with the one possible exception first.
Below in Lemma \ref{IIIlk}, we show that a Type III exception cannot occur for $l<k$, so we
don't need to worry about the extra decrease from the bias ratio.
But the one exception gives us knowledge of only one node $a$ in $L$. Suppose $a=1$.
Then all trees $t$ with $a=1$ will have $p_{tree}(t)=2$ and all trees $t$ with
$a=0$ will have $p_{tree}(t)=0$. Since each $p_{cat}$ involves an equally weighted
average of of trees $t$ with $a=0$ and $a=1$, each $p_{cat}$ will be set to 1 when $a$
is learned, so $S$ is unchanged when $a$ is learned and remains $2\tilde{n}$. Otherwise,
we still have the reduction in $S$ due to the $K_j$ sets as above, except with
now only one interval, so we obtain:
\begin{equation}
\Delta\log S\lesssim 4(c_{k-1,l}\beta+\dots+c_{k-1,1}\beta^{l})
\end{equation}
By induction we can plug in values for $c_{k-1,j}$, and using the asymptotic
relation, we find at the end of the computation, $S\gtrsim 2\tilde{n}$
\end{proof}

\begin{lemma} \label{IIIlk}
Type III exceptions can only occur when $l=k$.
\end{lemma}

\begin{proof}
We will show that when $l<k$, the bias ratio cannot be smaller than $2/\beta^3$.
Suppose $l<k$, and we are concerned that a Type III exception might happen at a subtree
rooted at $a\in L$, where $g$ is the root, due to measuring a leaf $b$.
 Let $q_r'$ be the probability that $v(a)=r$
based on all evaluated leaves descending from $g$. We will show
 $q_r'$ is close to $1/2$ before $b$ is evaluated, and therefore that the bias ratio is close to 1.

We can write:
\begin{align}
q_0'&=N\sum_{t:v(t,a)=0}w_tp_{tree}(t)\notag \\
&=N\sum_{t:v(t,a)=0}w_t\prod_{a'\in L}p_{root}(a',v(t,a')),
\end{align}
where $N$ is a normalization factor to ensure $q_0'+q_1'=1$, and $w_t$ is the
probability associated with $t$ in the distribution $\mathcal{T}_1$. $q_1'$ can
be written similarly. Note there
can be no exception giving with high probability the value of any other
$a'\in L$ before $b$ is measured,
as if there were, we couldn't have a
Type III exception at $a$, since the exception at $a$ would simply become a
Type I exception. Also, since $l<k$, less than
$\beta^{k-1}$ leaves have been evaluated, we can apply Lemma \ref{Alemm}
to all nodes $a'\in L$.

Using the sets $K_j$ from Lemma \ref{lemmaS}, we have that the number of
nodes in $L$ that have $\beta^j$ leaves evaluated is less than $\beta^{l-j+1}$.
Each of those nodes can have $p_{root}(a',r)$ that is at most at most
$1+c_{k-1,j}$ and at least $1-c_{k-1,j}$. If any node $a'\in L$ has not
had any evaluations, then $p_{root}(a',r)=1$.

Then
\begin{align}
q_0'<N\sum_{t:v(t,a)=0}w_t\prod_{j=1}^l(1+c_{k-1,j})^{\beta^{l-j+1}}\lesssim N/2,
\end{align}
and
\begin{align}
q_0'>N\sum_{t:v(t,a)=0}w_t\prod_{j=1}^l(1-c_{k-1,j})^{\beta^{l-j+1}}\gtrsim N/2,
\end{align}
So up to negligible factors, $q_0'=N/2$. The same reasoning holds for $q_1'$, and
using the normalization $q_1'+q_0'=1$, we have
$q_0'=q_1'=1/2$. Then the bias ratio is (from Eq. \ref{bratio})
\begin{align}
\frac{p_{partial(g)}}{p_{partial(a)}}=
\frac{p_{b|0}q_0'+
p_{b|1}q_1'}{(p_{b|0}+p_{b|1})/2}=1
\end{align}
up to negligible factors.
\end{proof}

\subsubsection{Bounding $D$}

To prove the lower bound on $D$, we reorder so that evaluations of leaves
descending from a single node in $L$ are grouped together (there are no intervals
this time). We are allowed to do this because $D$ only depends on the final
set of leaves that are evaluated, so if we reorder and find a bound on $D$
with the new ordering, it must also be a bound on any other ordering.
Furthermore, none of the Lemmas we use in proving the bound on $D$, or the probability of
exception when determining $D$
depend on ordering, so we can reorder freely.

We put evaluations for leaves descending from the at most $\beta$
sure nodes in $L$ first (the sure nodes are those obtained either by the one
exception on a subtree or by being given the value by the algorithm after
$\beta^{k-1}$ nodes of a subtree are evaluated). We call Phase I the evaluation
of leaves descending from those sure nodes, and we call Phase II the evaluation
of leaves after Phase I. Phase I evaluations exactly mimic the $k=1$ case for
bounding $D$, where the last evaluation to each node $a\in L$ can be thought
of as a query of that node. Thus in Phase I, the evolution of $p_{cat}(i,r)$
follows the $k=1$ case, except for the extra factors $F_j$ described above.

From this point, we will only discuss what happens in Phase II. Phase II
evaluations don't lead to high confidence knowledge of nodes in $L$, but we
still need to bound how much confidence can be obtained. For each $a\in L$
whose leaves are evaluated in Phase II, let $p_{node|i,r}(a)$ be the
probability that a tree remaining in $\mathcal{T}_{k,r,i}$ will have $v(a)=0$.
Analogously to the $k=1$ case, at the beginning of Phase II, for each $a$
there is a height $h$ that is the distance between $a$ and sure nodes, such
that $p_{node|i,r}(a)=1/2$ for $i>h$, and $p_{node|i,r}(a)=\{0,1\}$ for
$i+n_0\leq h$. If $p_{node|i,r}(a)=\{0,1\}$, it won't change with further
evaluations. However, we must bound how much $p_{node|i,r}(a)$ can change
if at the beginning of Phase II it is equal to $1/2$.

First we will bound how much $p_{node|i,r}(a)$ can change due to evaluations
on other nodes $a'\in L$ during Phase II, where $p_{node|i,r}(a)$ is initially
$1/2$. As in Lemma \ref{lemmaS}, we group $a'$ into the sets $K_j$. For any
tree $t\in \mathcal{T}_1$, and $a'\in K_j$, then
$\Delta\log p_{tree}(t)\leq 2c_{k-1,j}$ by Lemma \ref{Alemm}. Since we are
considering the case where $p_{node|i,r}(a)=1/2$ initially, half of the trees
$t$ with category $i$ and root $r$ have $a=0$ and half have $a=1$. Since each
$a'$ can change $p_{tree}$ by at most $1\pm2c_{k-1,j}$, the worst case is when
all for all trees $t$ with $a=0$, their probabilities change by $1+2c_{k-1,j}$,
and all trees with $a=1$ change by $1-2c_{k-1,j}$, for each $a'\in K_j$. Thus
we have $\Delta\log p_{node|i,r}(b)\leq 4c_{k-1,j}$. Summing over the sets
$K_j$, we obtain

\begin{equation}
\Delta\log p_{node|i,r}(a)\lesssim4(c_{k-1,l}\beta+c_{k-1,l-1}\beta^2+\dots+c_{k-1,1}\beta^l)
\end{equation}

\noindent where $l$ is strictly less than $k$, since we are in Phase II. This means
\begin{equation} \label{pnodediff}
|p_{node|i,r}(a)-\frac{1}{2}|\lesssim4(c_{k-1,l}\beta+c_{k-1,l-1}\beta^2+\dots+c_{k-1,1}\beta^l).
\end{equation}
We will use this result in calculating $D$.

To compute $D$, we need to track the evolution of $p_{cat}(i,r)$ in Phase II.
We will show that we can do this in terms of $p_{node|i,r}(a)$ and
$p_{root}(a,r)$. Let $p_{cat}(i,r)_i$ be the value of $p_{cat}(i,r)$ before any
leaves descending from $a$ are evaluated, and let $p_{cat}(i,r)_f$ be the value
of $p_{cat}(i,r)$ after all leaves descending from $a$ are evaluated. Then

\begin{align} \label{catdev}
p_{cat}(i,r)_f=& p_{cat}(i,r)_i(p_{node|i,r}(a)p_{root}(a,0)\notag \\
&+(1-p_{node|i,r}(a))p_{root}(a,1),
\end{align}

\noindent where $p_{root}(a,0)$, $p_{root}(a,1)$, and $p_{node|i,r}(a)$ are the values
after all leaves descending from $a$ are evaluated. By induction, we have
$|p_{root}(a,r)-1|\lesssim c_{k-1,j}$ if $j$ leaves of $a$ are evaluated.

Now we will see how $p_{cat}(i,r)$ is updated in the three cases:

\begin{itemize}

\item $\bf{[i+n_0\leq h]}$ Here we know $p_{node|i,0}(a)$ and $p_{node|i,1}(a)$
are either both $0$ or both $1$. Then by Eq. (\ref{catdev}), and using the bound
on $p_{root}(a,r)$ from Eq. (\ref{newconf}), $p_{cat}(i,r)$ will be multiplied
by at most $1+c_{k-1,j}$.
\item $\bf{[i>h]}$ Now there can be a difference between $p_{node|i,0}(a)$ and
$p_{node|i,1}(a)$, but this difference in bounded by Eq. (\ref{pnodediff}).
We have the bound on $p_{root}(a,r)$ from Eq. (\ref{newconf}). Plugging into
Eq. (\ref{catdev}), and using the fact that $p_{cat}(i,r)_i\leq1$, we obtain
$p_{cat}(i,r)_f$ as multiplying $p_{cat}(i,r)_i$ by at most $1+c_{k-1,j}$,
and then adding a quantity
$\gamma\lesssim 16c_{k-1,j}(c_{k-1,l}\beta+c_{k-1,l-1}\beta^2+\dots+c_{k-1,1}\beta^l)$.
Since the added term is from a single $i$, the addition to $D$ from all $i$'s
is $\lesssim \tilde{n}16c_{k-1,j}(c_{k-1,l}\beta+c_{k-1,l-1}\beta^2+\dots+c_{k-1,1}\beta^l)$

\item $\bf{[i\leq h<i+n_0]}$ Without knowing the details of the direct function
being evaluated by the tree, we have no knowledge of $p_{node|i,r}(a)$.
However, $|p_{root}(a,0)-p_{root}(a,1)|\lesssim 2c_{k-1,j}$, so plugging into
Eq. (\ref{catdev}) and using the fact that $p_{cat}(i,r)_i\leq1$ and
$p_{node|i,r}(a)\leq1$, we get that this step adds a quantity
$\gamma\lesssim 2n_0c_{k-1,j}$ to $D$.
\end{itemize}

Now we need to sum these contributions for all $a\in L$ in Phase II. This
summation simply expands each of the $c_{k-1,j}$ in the above cases into
$(c_{k-1,l}\beta+c_{k-1,l-1}\beta^2+\dots+c_{k-1,1}\beta^l)$. Putting the additions before
the multiplications, if $D_i$ is $D$ at the beginning of Phase II, and
$D_f$ is $D$ at the end of Phase II, we have
\begin{align}
D_f\lesssim&[D_i+16\tilde{n}(c_{k-1,l}\beta+c_{k-1,l-1}\beta^2+\dots+c_{k-1,1}\beta^l)^2\notag\\
&+2n_0(c_{k-1,l}\beta+c_{k-1,l-1}\beta^2+\dots+c_{k-1,1}\beta^l)]\notag\\
&\times(1+(c_{k-1,l}\beta+c_{k-1,l-1}\beta^2+\dots+c_{k-1,1}\beta^l))
\end{align}
Using our inductive assumption to plug in for $c_{k-1,l-i}$ from Lemma \ref{mainlem}, we
find that each term $c_{k-1,l-i}\beta^{i+1}\lesssim 8n_0\tilde{n}^{-3/5}\beta^7$.
So, aside from $D_i$, the largest term in the equation scales like $\tilde{n}^{-1/5}$,
which goes to 0 for large $\tilde{n}$.

For the case $k=l$, from the $k=1$ case, we have $D_i\lesssim \Pi_jF_jn_0\beta$. (The $F_j$
comes from renormalizing $p_{root}$ in Eq. \ref{eqtree}). So
$D_ i$ dominates the sum, (since the $F_j$ are all larger than 1),
 and thus $D_f\lesssim \Pi_jF_jn_0\beta$. We obtain,
using our bound from Section \ref{Sbound},
\begin{equation}
c_{k,k}\leq\frac{D}{S}\lesssim\frac{2n_0\beta^4}{\tilde{n}^{3/5}}.
\end{equation}

For the case $l<k$, we can get at most 1 sure node in $L$ from one exception
in a subtree. In that case, Phase I can never get started, so we go straight
into Phase II. Thus $D_i=0$, and the above sum is dominated by the terms
$16\tilde{n}(c_{k-1,l}\beta)^2$ and $2n_0c_{k-1,l}\beta$. From our bound on $S$
from Section \ref{Sbound}, we have $S\gtrsim 2\tilde{n}$. Combining these two
bounds, we obtain
\begin{equation}
c_{k,l}\lesssim(8n_0)^{k-l+1}\tilde{n}^{-3(k-l+1)/5}\beta^{4+6(k-l)}.
\end{equation}
This proves the $c_{k,l}$ bound of Lemma \ref{mainlem}

\end{document}